\documentclass[letterpaper, 10 pt, conference]{ieeeconf}  
\IEEEoverridecommandlockouts                              
\usepackage{epsfig} 

\usepackage{amsmath,amssymb,array,graphicx,verbatim}

\usepackage{amsthm}
\usepackage{blkarray}
\usepackage{relsize}
\usepackage{mathtools}
\usepackage{enumerate}

\newtheorem{theorem}{Theorem}
\newtheorem{lemma}{Lemma}
\newtheorem{problem}{Problem}

\newtheorem{remark}{Remark}
\usepackage{enumerate}
\usepackage{epstopdf}
\usepackage[noadjust]{cite}
  
\usepackage{dblfloatfix}
\usepackage{dsfont}
\usepackage{mathrsfs}

\usepackage{url}

\usepackage [english]{babel}
\usepackage [autostyle, english = american]{csquotes}
\MakeOuterQuote{"}
\usepackage{tikz}
\usepackage{circuitikz}
\usepackage{tikz-cd}
\usetikzlibrary{arrows,shapes,calc,positioning}
\tikzstyle{block} = [draw,rectangle, rounded corners, minimum width=1cm, minimum height=0.8cm,text centered, line width=2pt ]
\tikzstyle{arrow} = [thick,->,>=stealth,line width=2pt]

\tikzset{cross/.style={cross out, draw=black, minimum size=2*(#1-\pgflinewidth), inner sep=0pt, outer sep=0pt},
cross/.default={1pt}}

\usepackage{enumerate}
\tikzset{
  shift left/.style ={commutative diagrams/shift left={#1}},
  shift right/.style={commutative diagrams/shift right={#1}}
}
\usetikzlibrary{calc}

\newcommand\rsmraise[1]{%
  \ifx#1\displaystyle .8\else
    \ifx#1\textstyle .8\else
      \ifx#1\scriptstyle .6\else
        .45%
      \fi
    \fi
  \fi}

\usepackage{epstopdf}
\usetikzlibrary{shapes}
\tikzstyle{block} = [draw,rectangle, rounded corners, minimum width=1cm, minimum height=0.8cm,text centered, line width=2pt ]
\tikzstyle{arrow} = [thick,->,>=stealth,line width=2pt]

\usetikzlibrary{arrows,decorations.markings,decorations}
\tikzset{
    addarrow/.style={decoration={markings, mark=at position 1 with {\arrow{stealth}}},
                     postaction={decorate}}
}

\usetikzlibrary{automata,arrows,positioning,calc}

\newcommand{\tp}{\intercal}		
\newcommand{\R}{\mathbb{R}}			

\newcommand{\ind}{\mathds{1}}		

\DeclareMathOperator{\ee}{\mathbb{E}}			
\DeclareMathOperator{\prob}{\mathbb{P}}			
\DeclareMathOperator{\vecc}{\mathbf{vec}}		
\DeclareMathOperator{\tr}{\mathbf{tr}}			
\DeclareMathOperator{\cov}{\mathbf{cov}}		

\title{\LARGE \bf
Decentralized Control of Stochastically Switched Linear System with Unreliable Communication
}

\author{Seyed Mohammad Asghari, Yi Ouyang, and Ashutosh Nayyar
\thanks{
S. M. Asghari and A. Nayyar are with the Department of Electrical
Engineering, University of Southern California, Los Angeles, CA.
Y. Ouyang is currently with the University of California, Berkeley, CA. 
Email: asgharip@usc.edu; ouyangyi@berkeley.edu; ashutosn@usc.edu.}
 \thanks{This research was supported by NSF under grants ECCS 1509812 and CNS 1446901.}
}

\begin{document}

\maketitle
\thispagestyle{empty}
\pagestyle{empty}



{\color{black}
\begin{abstract}
We consider a networked control system (NCS) consisting of two plants, a global plant and a local plant, and two controllers, a global controller and a local controller. 
The global (resp. local) plant follows discrete-time stochastically switched linear dynamics with a continuous global (resp. local) state and a discrete global (resp. local) mode.
We assume that the state and mode of the global plant are observed by both controllers while the state and mode of the local plant are only observed by the local controller. The local controller can inform the global controller of the local plant's state and mode through an unreliable TCP-like communication channel where successful transmissions are acknowledged. The objective of the controllers is to cooperatively minimize a modes-dependent quadratic cost over a finite time horizon. Following the method developed in \cite{Ouyang_Asghari_Nayyar:2016} and \cite{Asghari2016optimal}, we construct a dynamic program based on common information and a decomposition of strategies, and use it to obtain explicit optimal strategies for the controllers. In the optimal strategies, both controllers compute a common estimate of the local plant's state. The global controller's action is linear in the state of the global plant and the common estimated state, and the local controller's action is linear in the actual states of both plants and the common estimated state. Furthermore, the gain matrices for the global controller depend on the global mode and its observation about the local mode, while the gain matrices for the local controller depend on the actual modes of both plants and the global controller's observation about the local mode.
\end{abstract}
}

\section{Introduction}

The widespread applications of decentralized control in networked control systems (NCSs), power systems, smart buildings,  autonomous vehicles and economic models  have made it a topic of interest in the recent years
\cite{HespanhaSurvey, jadbabaie2003coordination, zampolli2006optimal}. Despite increased efforts in advancing this field, decentralized control still remains challenging with a broad range of problems to be investigated. One class of such problems is the control of switched systems where switching is governed by stochastic parameters that are partially observed by the individual decision makers. The stochastic parameters can represent local and global system conditions like the state of communication links among the controllers, mission objectives, physical conditions, and changes in the system environment. 

While centralized control of switched systems has been extensively addressed by several studies \cite{tabuada2009verification, borrelli2003constrained, costa1995discrete, chizeck1988optimal,sinopoli2004kalman,imer2006optimal}, the  decentralized counterpart has been the focus of relatively sparse studies  \cite{xiong2009local,Farokhi_TAC_2015, mishra2015team}. This is due to the fact that decentralized control problems are generally difficult to solve (see \cite{Witsenhausen:1968, LipsaMartins:2011b,blondel2000survey}). Not only are  linear strategies suboptimal for such problems, but the problem of finding the best linear strategies may not be convex \cite{YukselBasar:2013}. Existing methods for solving decentralized control problems require either specific information structures, such as static \cite{Radner:1962}, partially nested \cite{HoChu:1972,LamperskiDoyle:2011,LessardNayyar:2013,ShahParrilo:2013, Nayyar_Lessard_2015,Lessard_Lall_2015}, stochastically nested \cite{Yuksel:2009}, switched partially nested \cite{Ouyang_Asghari_Nayyar:2017} or other specific properties, such as quadratic invariance \cite{RotkowitzLall:2006} or substitutability \cite{AsghariNayyar:2015, asghari2016dynamic} which make the decentralized control problems "simpler" than the general problems. 

{\color{black}
In this paper, we consider a NCS with discrete-time stochastically switched linear dynamics.
The NCS includes two plants, called "global" plant and "local" plant, and two controllers, namely global controller $C^0$ and local controller $C^1$ as shown in Fig. \ref{fig:SystemModel}. 
Associated with the global (resp. local) plant, there is a continuous global (resp. local) state and a discrete global (resp. local) mode.
The discrete global (resp. local) mode allows us to model non-linear dynamics such as abrupt environmental disturbances, component failures or repairs, changes in subsystems interconnections, and abrupt changes in the operation point \cite{costa2006discrete}.

We assume that the control action of the global controller $C^0$ can affect both plants while the control action of the local controller $C^1$ can only affect the local plant as its name suggests. 
We assume that the state and mode of global plant are observed by both controllers while the state and mode of local plant is only observed by the local controller $C^1$. In addition to the observations, controller $C^1$ can inform controller $C^0$ of the local plant's state and mode through a communication channel with random packet drops. We assume a TCP-like protocol \cite{Schenato2007} where successful transmissions of packets are acknowledged by controller $C^0$. The objective of the controllers is to cooperatively minimize a modes-dependent quadratic cost over a finite time horizon. The dependence of cost on the global and local modes allows us to model mode-dependent changes in the control objective.

This stochastically switched system model can arise in various NCS applications. For example, the operation of a service robot in a smart home can be modeled by the NCS where the robot is the local plant and the smart home is the global plant. The global controller can be the HVAC (heating, ventilation, and air conditioning) system that controls the global state and mode including temperature and indoor air quality, while the local controller directly controls the robot with switched linear dynamics. Depending on the local situation, the robot can transmit information through wireless communication to the global controller to request adjustments on certain comfort parameters of the smart home.
}

The  decentralized control problem we consider in this paper does not belong to the simpler classes mentioned earlier due to either the unreliable communication or the switching dynamics and cost function. A closely related problem has been studied in \cite{mishra2015team}; however; the information structure there is partially nested. We developed a method to obtain optimal controllers for a decentralized control problem with unreliable communication in \cite{Ouyang_Asghari_Nayyar:2016, Asghari2016optimal,Ouyang_Asghari_Nayyar:2018} using ideas from the common information approach \cite{nayyar2013decentralized}. More specifically, we proposed a modified dynamic program based on a decomposition of strategies which  can be explicitly solved to find the optimal controllers for the problems of \cite{Ouyang_Asghari_Nayyar:2016, Asghari2016optimal}. Although this method fails to provide any structure for the dynamic program of the switched system considered in this paper, we  show that it can be generalized to capture the switching nature of system dynamics of the plants. Using this method, we obtain explicit optimal strategies for the controllers.  In the optimal strategies, both controllers compute a common estimate of the local plant's state. Global controller $C^0$'s action is linear in the state of global plant and the common estimated state, and the local controller $C^1$'s action is linear in both the actual states of plants and the common estimated state. Furthermore, the gain matrix for controller $C^0$ depends on the global mode and its observation about the local mode, while the gain matrices for controller $C^1$ depend on the actual global and local modes of both plants and controller $C^0$'s observation about the local mode. 

\subsection{Notation}
Random variables/vectors are denoted by upper case letters, their realization by the corresponding lower case letters.
For a sequence of column vectors $X, Y, Z,...$, the notation $\vecc(X,Y,Z,...)$ denotes vector $[X^{\tp}, Y^{\tp}, Z^{\tp},...]^{\tp}$. The transpose and trace of matrix $A$ are denoted by $A^{\tp}$ and $\tr(A)$, respectively. 
Uppercase letters are also used to denote matrices. Matrix  dimensions are not specified if they can be inferred from the context.
The notation $\mathbf{I}_{n}$ and $\mathbf{0}_{n \times m}$ is used to denote a $n \times n$ identity matrix and a $n \times m$ zero matrix, respectively. For block matrix $B$ and $r=0,1$, $[B]_{r \bullet}$ denotes the $r$-th block row of $B$. Note that the top block row corresponds to $r=0$ and the bottom corresponds to $r=1$.  For example, for 
$B = \begin{bmatrix}
\mathbf{I}_{m} &\mathbf{0}_{m \times n} \\ \mathbf{0}_{n \times m}  &\mathbf{I}_{n}
\end{bmatrix}$, $[B]_{0 \bullet}= [\mathbf{I}_{m}\hspace{3mm} \mathbf{0}_{m \times n}]$ and $[B]_{1 \bullet}= [\mathbf{0}_{n \times m} \hspace{3mm} \mathbf{I}_{n}]$.

In general, subscripts are used as time index while superscripts are used to index controllers.
For time indices $t_1\leq t_2$, $X_{t_1:t_2}$ (resp. $g_{t_1:t_2}(\cdot)$) is the short hand notation for the variables $(X_{t_1},X_{t_1+1},...,X_{t_2})$ (resp.  functions $(g_{t_1}(\cdot),\dots,g_{t_1}(\cdot))$). 
{\color{black} Similarly, $X^{0:1}$ is the shorthand notation for the variables $X^0,X^1$.}
$\prob(\cdot)$, $\ee[\cdot]$, and $\cov(\cdot)$ denote the probability of an event, the expectation of a random variable/vector, and the covariance matrix of a random vector, respectively.
For random variables/vectors $X$ and $Y$, $\ee[X|y] := \ee[X|Y=y]$. 
For a strategy $g$, we use $\prob^g(\cdot)$ (resp. $\ee^g[\cdot]$) to indicate that the probability (resp. expectation) depends on the choice of $g$. 
Let $\Delta(\R^n)$ denote the set of all probability measures on $\R^n$ with finite second moment. 
For any $\theta \in \Delta(\R^n)$, $\theta(E) = \int_{\R^n} \mathds{1}_{E}(x) \theta(dx)$ denotes the probability of event $E$ under $\theta$. 
The mean and the covariance of a distribution $\theta \in \Delta(\R^n)$ are denoted by $\mu(\theta)$ and $\cov(\theta)$, respectively, and are defined as $\mu(\theta) = \int_{R^n} x \theta(dx)$ and 
$\cov(\theta) = \int_{R^n} (x - \mu(\theta)) (x - \mu(\theta))^{\tp} \theta(dx)$.  

\subsection{Organization}
The rest of the paper is organized as follows. We introduce the system model and formulate the two-controller optimal control problem in Section \ref{sec:model}. In Section \ref{sec:equivalent_prob}, following the common information approach, we  provide a dynamic program for solving this problem. 
Section \ref{sec:decomposition_DP} introduces a decomposition for the control strategies and provides a modified dynamic program based on this decomposition. We solve the modified dynamic program in Section \ref{sec:solution} and provide the optimal strategies for the controllers. Section \ref{sec:conclusion} concludes the paper. The proofs of all the technical results of the paper appear in the Appendices.

\section{System Model and Problem Formulation}
\label{sec:model}

Consider the discrete-time switched linear system with two plants and two associated controllers shown in Fig. \ref{fig:SystemModel}.
The two-plant system follows the switched linear dynamics described by
\begin{align}
&\begin{bmatrix}
X_{t+1}^0 \\ X_{t+1}^1
\end{bmatrix} = A(M_t^{0:1})\begin{bmatrix}
X_{t}^0 \\ X_{t}^1
\end{bmatrix}
+B(M_t^{0:1})\begin{bmatrix}
U_{t}^0 \\ U_{t}^1
\end{bmatrix}+\begin{bmatrix}
W_{t}^0 \\ W_{t}^1
\end{bmatrix}
 \label{Model:system_0}
\end{align}
for $t=0,\dots,T$ where $X_t^n \in \R^{d_X^n}$ is the continuous state and $M^n_t \in \mathcal{M}^n = \{1,\ldots, \kappa^n \}$ with $\kappa^n <\infty$ is the discrete mode of plant $n$ for $n=0,1$. $U_t^n \in \R^{d_U^n}$ is the control action of controller $C^n$ and $W_t^n$ is a zero mean noise vector.
We assume that the collection of random variables $X_0^{n}, W_{t}^{n},M_{t}^{n}$ for $n=0,1, t=0,1,\dots,T,$ are independent and with distributions 
$\pi_{X_0^n}$ and $\pi_{W_t^n},\pi_{M^n}$, respectively. We further assume that $X_0^{0:1},W_{0:T}^{0:1}$ have finite second moments.

\begin{figure}
\begin{center}
\begin{tikzpicture}
\node [rectangle,draw,minimum width=1.2cm,minimum height=1cm,line width=1pt,rounded corners]at (-1.4,0) (1) {
\begin{small}
\begin{tabular}{c}
Global  \\ Controller $C^0$
\end{tabular}
\end{small}}; 
\node [rectangle,draw,minimum width=1.2cm,minimum height=1cm,line width=1pt,rounded corners]at (4.45,0) (2) {
\begin{small}
\begin{tabular}{c}
Local \\ Controller $C^1$
\end{tabular}
\end{small}
}; 
\node [rectangle,draw,minimum width=1.5cm,minimum height=1cm,line width=1pt,rounded corners]at (-1,2.5) (3){\begin{small}
\begin{tabular}{c}
Global Plant 0 \\ $X_t^0, M_t^0$
\end{tabular}
\end{small}};
\node [rectangle,draw,minimum width=1.5cm,minimum height=1cm,line width=1pt,rounded corners]at (4,2.5) (4) {\begin{small}
\begin{tabular}{c}
Local Plant 1 \\ $X_t^1, M_t^1$
\end{tabular}
\end{small}}; 



\draw [line width=1pt] (2.5,0) to[cspst] (0.5,0);
\path[thick,->,>=stealth,dashed,line width=1pt]
           (4,0.5) edge node {}   (4,2.0)           
           (-1,0.5) edge node {}   (-1,2.00)  ;
                      \path[thick,->,>=stealth,dash dot,line width=1pt]
  ($(3.east)$) edge node{} ($(4.west)$);
    \path[thick,->,>=stealth, shift left=.30ex,line width=1pt]
        (3.7,2.0) edge node {}   (3.7,0.5)
         (-0.05,1.95) edge node{} (3.65,0.5)
         (-1.3,2.0) edge node {}   (-1.3,0.5) ;
\path[thick,->,>=stealth, line width=1pt]
 (0.5,-0) edge node {} (-0.1,-0)
(4.0,-1) edge node {} (4.0,-0.5);
 
\path[thick, line width=1pt]
 (3.1,-0.0) edge node {} (2.5,-0.0)
 (0.5,-0) edge node {} (0.5,-1)
  (0.5,-1) edge node {} (4,-1);
 
         
         
\node[] at (4.3,1) {$U_t^1$};
\node[] at (-0.7,1) {$U_t^0$};
\node[] at (1.6,-0.5) {$\Gamma_t$};
\node[] at (2.5,0.3) {$X_t^1,M_t^1$};
\node[] at (0.4,0.3) {$Z_t, \tilde Z_t$};

\draw[thick,->,>=stealth, dashed, line width=1pt] (1)--node[above, sloped, pos=.3]{$U^0_t$}(4);
\end{tikzpicture}
\caption{Two-controller system model. The binary random variable $\Gamma_t$ indicates whether packets are transmitted successfully. 
{\color{black}Solid lines indicate communication links, dashed lines indicate control links, and the dash-dot line indicates that the global state and mode can affect the local state.}}
\label{fig:SystemModel}
\end{center}
\end{figure}
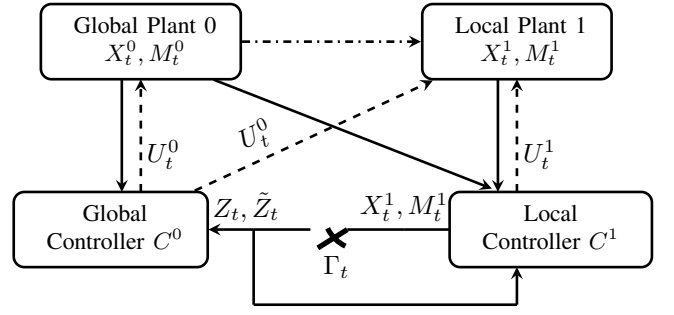

The system matrices $A(M_t^{0:1})$ and $B(M_t^{0:1})$ depend on the system modes $M_t^0$ and $M^1_t$.
As illustrated in Fig. \ref{fig:SystemModel}, state and mode of plant $0$ and the action of controller $C^0$ can affect plant $1$, but state and mode of plant $1$ and the action of controller $C^1$ do not affect plant $0$.

\vspace{-4mm}
\begin{small}
\begin{align}
A(M_t^{0:1}) = \begin{bmatrix}
A^{00}(M_t^{0}) & \mathbf{0}\\
A^{10}(M_t^{0:1}) & A^{11}(M_t^{0:1})
\end{bmatrix},
\\
B(M_t^{0:1}) = \begin{bmatrix}
B^{00}(M_t^{0}) & \mathbf{0}\\
B^{10}(M_t^{0:1}) & B^{11}(M_t^{0:1})
\end{bmatrix}.
\end{align}
\end{small}
Since $M^0_t$ affects both plants and $M^1_t$ only affects plant $1$, we refer to $M^0_t$ as the global mode and to $M^1_t$ as the local mode. 
For notational simplicity, we denote $X_t = \vecc({X_t^{0:1}}), U_t = \vecc({U_t^{0:1}}), W_t = \vecc({W_t^{0:1}})$, and $S_t = \vecc(X_t^{0:1},U_t^{0:1})$. Then, the system dynamics can be expressed as $X_{t+1} = D(M_t^{0:1}) S_t + W_t$ where

\vspace{-3mm}
\begin{small}
\begin{align}
D(M_t^{0:1}) =
\begin{bmatrix}
A(M_t^{0:1}) & B(M_t^{0:1})
\end{bmatrix}.
\label{D_matrix}
\end{align}
\end{small}
{\color{black}
At each time $t$, the state $X_t^0$ and mode $M_t^0$ of the global plant are observed by both controllers $C^0$ and $C^1$ while the state $X_t^1$ and mode $M_t^1$ of the local plant is only observed by controller $C^1$. Controller $C^1$ can inform controller $C^0$ of the local plant's state and mode through an unreliable link with random packet drops. }
Let $\Gamma_t$ be a Bernoulli random variable describing the state of this link, that is, $\Gamma_t=0$ when the link is broken and otherwise, $\Gamma_t=1$. We denote $p(i) := \prob(\Gamma_t = i)$ for $i \in \{0,1\}$. We assume that $\Gamma_{0:T}$ are independent random variables and they are independent of $X_0^{0:1},W_{0:T}^{0:1},M_{0:T}^{0:1}$. Furthermore, let $(Z_t,\tilde Z_t)$ be the channel output, that is,
 \begin{align}
 Z_t = &
\left\{\begin{array}{ll}
X_t^1 & \text{ when } \Gamma_t = 1,\\
\emptyset & \text{ when } \Gamma_t = 0.
\end{array}\right. \\
\tilde Z_t = &
\left\{\begin{array}{ll}
M_t^1 & \text{ when } \Gamma_t = 1,\\
\emptyset & \text{ when } \Gamma_t = 0.
\end{array}\right.
\label{Model:channel}
\end{align}

{\color{black}
We assume that the channel outputs $Z_t$ and $\tilde Z_t$ are perfectly observed by controller $C^0$.  The successful transmissions of packets is acknowledged by the controller $C^0$ (see, for example, TCP-like protocols \cite{Schenato2007}).
Thus, effectively, $Z_t$ and $\tilde Z_t$ are perfectly observed by controller $C^1$ as well. }
Both controllers select their control actions at time $t$ after observing $Z_t$ and $\tilde Z_t$.  We assume that the links from controllers $C^0$ and $C^1$ to the plants are perfect.

Let $H_t^n$ denote the information available to controller $C^n$, $n \in \{0,1\}$ to make decisions at time $t$. Then,
\begin{align}
H_t^0 &= \{X_{0:t}^0, M_{0:t}^0, Z_{0:t}, \tilde Z_{0:t}, U_{0:t-1}^{0}\}, \notag \\
H_t^1 &=  H_t^0 \cup \{X_{t}^{1}, M_{t}^{1}\}\footnotemark .
\label{Model:info}
\end{align}
\footnotetext{One can follow the argument of \cite[Lemma 1]{Ouyang_Asghari_Nayyar:2016} to see that the results of this paper hold for any $H_t^1 = H_t^0 \cup \hat H_t^1$ where 
$\{X_{t}^{1}, M_{t}^{1}\} \subseteq \hat H_t^1  \subseteq \{X_{0:t}^{1}, M_{0:t}^{1}, U_{0:t-1}^1\}$. For simplicity of presentation, we restrict to $\hat H_t^1 = \{X_{t}^{1}, M_{t}^{1}\}$.}
Let $\mathcal{H}^{n}_t$ be the space of all possible realizations of $H_t^n$, $n \in \{0,1\}$.
Then, $C^{n}$'s actions are selected according to
\begin{align}
U^{n}_t &= \lambda^{n}_t(H^{n}_t), \hspace{2mm}  n \in \{0,1\},
\label{Model:strategy}
\end{align}
where $\lambda^{n}_t:\mathcal{H}^{n}_t \to \R^{d_U^n}$ is a Borel measurable mapping.
The collection of mappings $\lambda_0^n, \ldots, \lambda_T^n$ is called the strategy of controller $C^n$ and is denoted by $\lambda^n$. The collection of both controllers' strategies, $\lambda^{0:1}$, is called the strategy profile.

The instantaneous cost $c_t(X_t^{0:1}, M_t^{0:1}, U_t^{0:1})$ of the system is a quadratic function given by
\begin{align}
c_t(X_t^{0:1}, M_t^{0:1}, U_t^{0:1}) &= 
X_t^\tp Q_t(M_t^{0:1}) X_t + U_t^\tp R_t(M_t^{0:1}) U_t, 
\label{cost_structure}
\end{align}
where $Q_t(m^{0:1})$ is a symmetric positive semi-definite (PSD) matrix and 
$R_t(m^{0:1})$ is a symmetric positive definite (PD) matrix for all $m^0 \in \mathcal{M}^0$ and $m^1 \in \mathcal{M}^1$.

The performance of strategies $\lambda^{0:1}$ is the total expected cost given by

\vspace{-4mm}
\begin{small}
\begin{align}
J(\lambda^{0:1})=\ee^{\lambda^{0:1}}\left[\sum_{t=0}^T c_t(X_t^{0:1}, M_t^{0:1}, U_t^{0:1}) \right].
\label{Model:obj}
\end{align}
\end{small}
Let $\Lambda^0$ and $\Lambda^1$ denote the set of all possible control strategies of $C^0$ and $C^1$, respectively, that ensure all random variables (state and control actions) have finite second moments. The optimal control problem is formally defined below.
\begin{problem}
\label{problem1}
For the system described by \eqref{Model:system_0}-\eqref{Model:obj}, 
we would like to solve the following strategy optimization problem,
\begin{align}
\inf_{\lambda^0\in\Lambda^0, \lambda^1\in \Lambda^1} J(\lambda^{0:1}).
\end{align}
\end{problem}

Problem \ref{problem1} is a two-controller decentralized optimal control problem. However, we cannot a priori assume that linear control strategies are optimal in this problem because 1) the information structure is not partially nested \cite{HoChu:1972}, 2) the system dynamics given in \eqref{Model:system_0} are not linear due to the presence of the local and global modes $M_t^{0:1}$, 3) $X_0^{0:1}$ and $W_{0:T}^{0:1}$ are
not necessarily Gaussian.


\section{Equivalent Problem and Dynamic Program}
\label{sec:equivalent_prob}


According to \eqref{Model:info}, $H_t^0$ is the \emph{common information} among the controllers $C^0$ and $C^1$. Using the common information approach \cite{nayyar2013decentralized}, we formulate a centralized decision-making problem 
which can be used to find optimal strategies in the decentralized Problem \ref{problem1}. In this centralized problem, $C^0$ is the only decision-maker and at each time $t$, it makes two decisions given the realization $h_t^0$:
\begin{enumerate}
\item $C^0$'s control action $u_t^0 = \phi_t^0 (h_t^0)$,
\item A prescription $\rho_t$ for $C^1$ which is a Borel measurable mapping from $\R^{d_X^1} \times \mathcal{M}^1$ to $\R^{d_U^1}$. That is, $\rho_t = \phi_t^{1}(h_t^0)$ where $\rho_t$ belongs to the space $\mathcal{P} =\{\rho: \R^{d_X^1} \times \mathcal{M}^1 \to \R^{d_U^1} \text{ such that $\rho$ is measurable} \}$.
\end{enumerate}
For each realization of $C^1$'s private information $(x_t^1,m_t^1)$, the mapping $\rho_t$ prescribes an action $u_t^1 = \rho_t(x_t^1,m_t^1)$.

We call $u_t^{prs} = (u_t^0, \rho_t)$ the prescription at time $t$. We denote $\phi_{t}^{prs} = (\phi^0_{t}, \phi^{1}_{t})$ and write $u^{prs}_t = \phi_{t}^{prs}(h^0_t)$ 
to indicate that the prescription is a function of the common information $h^0_t$.
The functions $(\phi_{t}^{prs},  \hspace{2pt} 0 \leq t \leq T)$ are collectively referred to as the prescription strategy and denoted by $\phi^{prs}$.
The prescription strategy is required to satisfy the following conditions: 1) $\phi^0 \in \Lambda^0$, and 2) if we define $\varphi_t^{1}(X_t^1, M_t^1, H_t^0):= [\phi_t^{1}(H_t^0)](X_t^1,M_t^1) $, then $\varphi^{1} \in \Lambda^{1}$ where the notation $[\phi_t^{1}(H_t^0)](X_t^1,M_t^1)$ means that we first use $\phi_t^1(H^0_t)$ to find the mapping $\rho_t$ and then evaluate $\rho_t$ at $X_t^1,M_t^1$.

Denote by $\Phi^{prs}$ the set of all prescription strategies satisfying the above conditions. Consider the following optimization problem.

\begin{problem}[Equivalent Centralized Problem]
\label{problem:equivalent}
For the system described by \eqref{Model:system_0}-\eqref{Model:strategy}, we would like to solve the following optimization problem,

\vspace{-4mm}
\begin{small}
\begin{align}
\inf_{\phi^{prs} \in \Phi^{prs}} \ee^{\phi^{prs}} \left[\sum_{t=0}^T c_t^{prs}(X_t^{0:1}, M_t^{0:1}, U^{prs}_t)\right],
\label{Model:optimization_equivalent}
\end{align}
\end{small}
where for any $x_t^{0:1}$, $m_t^{0:1}$, and $u^{prs}_t=(u^0_t,\rho_t)$,
\begin{small}
\begin{align}
c_t^{prs}(x^{0:1}_t, m^{0:1}_t, u^{prs}_t) = c_t\Big(x_t^{0:1}, m_t^{0:1}, u^{0}_t, \rho_t(x_t^1, m_t^1)\Big)  .
\label{cost_of_new_problem}
\end{align}
\end{small}
\end{problem}
Using arguments similar to \cite{Asghari2016optimal}, the solution to Problem \ref{problem:equivalent} can be characterized by a Dynamic Program (DP). The information state for this DP consists of $x_t^0,m_t^0,\tilde z_t$, and the common belief on the state $X_t^1$ defined as follows: Under prescription strategies $\phi^{prs}_{0:t-1} \in \Phi^{prs}$ and for any measurable sets $E \subset \R^{d_X^1}$, 
\begin{align}
\theta_t(E) := \prob^{\phi^{prs}_{0:t-1}}(X_t^1 \in E | h^0_{t}).
\label{eq:thetat_margin}
\end{align}
This belief is sequentially updated for any realization $h^0_t$ according to
\begin{align}
\theta_{t+1} = \psi_t(\theta_t, u_t^{prs}, x_t^0, m_t^0, \tilde z_t, z_{t+1}),
\label{update_theta}
\end{align}
where $u_t^{prs} = (u_t^0, \rho_t)= \phi^{prs}_{t}(h^0_t)$ and $\psi_t$ is a fixed transformation which does not depend on the choice of prescription strategies\footnote{See Appendix \ref{Proof_Lm_belief_update} for the exact description of transformation $\psi_t$.}.

Let $\tilde{\mathcal{M}}^1 = \mathcal{M}^1 \cup \{\emptyset\}$. Then, the following theorem provides a DP characterization for the solution of Problem~\ref{problem:equivalent}. 
\begin{theorem}
\label{thm:general_structure}
Suppose there exist functions $\{ V_{t}: \R^{d_X^0} \times  \mathcal{M}^0 \times \Delta(\R^{d_X^1}) 
\times \tilde{\mathcal{M}}^1 \to \R \text{ for }t=0,1,\dots,T+1$ such that for each $x_t^0 \in \R^{d_X^0}$, $m_t^0 \in \mathcal{M}^0$, $\theta_t \in \Delta(\R^{d_X^1})$, and $\tilde z_t \in \tilde{\mathcal{M}}^1$, the following are true:
\begin{itemize}
\item $V_{T+1}(x_{T+1}^0, m_{T+1}^0, \theta_{T+1}, \tilde z_{T+1})= 0$,
\item For any $t = 0,1,\dots,T$
\begin{small}
\begin{align}
& \hspace{-4mm} V_t(x_{t}^0, m_{t}^0, \theta_{t}, \tilde z_{t}) =   
 \min
 \Big\lbrace
\ee \Big[ c_t^{prs}(x_t^0, X_t^{1}, m_t^0, M_t^{1}, u^{prs}_t) \notag \\
&+ V_{t+1}\big(X_{t+1}^0, M_{t+1}^0, \psi_t(\theta_t, u_t^{prs}, x_t^0, m_t^0, \tilde z_t, Z_{t+1})
, \tilde Z_{t+1}\big) \notag \\
&\hspace{40mm}  \Big| x_t^0, m_t^0, \theta_t, \tilde z_t, u_t^{prs} \Big]
 \Big \rbrace,
\label{eq:DP_V}
\end{align}
\end{small}
where $X_t^1$ is a random vector distributed according to $\theta_t$, $u^{prs}_t=(u^{0}_t, \rho_t)$ and the minimization is over $u^0_t \in \R^{d_U^0}, \rho_t \in \mathcal{P}$.
\end{itemize}
Further, suppose there exists a feasible prescription strategy
$\phi^{prs*}\in \Phi^{prs}$
such that for any realization $h^0_t\in\mathcal{H}^0_t$ and its corresponding common beliefs $\theta_t$ (defined in \eqref{eq:thetat_margin}-\eqref{update_theta}), the prescription $u^{prs*}_t=(u^{0*}_t, \rho_t^*)= \phi^{prs*}(h^0_t)$
achieves the minimum in the definition of $V_t(x_{t}^0, m_{t}^0, \theta_{t}, \tilde z_{t})$. Then, $\phi^{prs*}$ is an optimal prescription strategy for Problem \ref{problem:equivalent}.
\end{theorem}

\begin{proof}
See Appendix \ref{Proof_Thm_general_structure} for a proof.
\end{proof}

\section{Modified Dynamic Program based on Strategy Decomposition}
\label{sec:decomposition_DP}

In this section, we first introduce a decomposition for the control strategies of the controller $C^1$. Then, we modify the DP of Theorem \ref{thm:general_structure} based on this decomposition. As we will show in Section \ref{sec:solution}, this decomposition helps us find a solution to the DP of Theorem \ref{thm:general_structure}, and provide optimal strategies for the controllers.

\subsection{Decomposition of controller $C^1$'s strategies}
\label{section:decomposition_local_strategies}
Consider arbitrary prescription strategies $\phi^{prs} \in \Phi^{prs}$ of Problem \ref{problem:equivalent}. Under these strategies, $U_t^1$ can be decomposed as,

\vspace{-4mm}
\begin{small}
\begin{align}
&U_t^{1} = [\phi_t^1(H_t^0)](X_t^1,M_t^1) \notag \\
&= \underbrace{\ee^{\phi_{0:t-1}^{prs}} [U_t^{1} \vert H_t^0,M_t^1]}_{ [\bar \phi_t^{1}(H_t^0)](M_t^1)}
+ \underbrace{\{ U_t^{1} - \ee^{\phi_{0:t-1}^{prs}} [ U_t^{1} \vert H_t^0,M_t^1 ]  \}}_{[\tilde \phi_t^{1}(H_t^0)](X_t^1,M_t^1)}.
\label{decomposition}
\end{align}
\end{small}
Note that for any realization $h_t^0$ of $H_t^0$, $\bar \phi_t^{1}(h_t^0)$ is a measurable mapping from $\mathcal{M}^1$ to $\R^{d_U^1}$ and $\tilde \phi_t^{1}(h_t^0)$ is a measurable mapping from $\R^{d_X^1} \times \mathcal{M}^1$ to $\R^{d_U^1}$. Furthermore, for any realization $h_t^0$ of $H_t^0$ and $m_t^1$ of $M_t^1$, $[\tilde \phi_t^{1}(H_t^0)](X_t^1,M_t^1)$ is conditionally zero-mean given $h_t^0,m_t^1$, that is, $\ee^{\phi_{0:t-1}^{prs}} \big[[\tilde \phi_t^{1}(H_t^0)](X_t^1,M_t^1)  \vert h_t^0,m_t^1\big] =0$. Also note that $X_t^1$ and $M_t^1$ are conditionally independent given the common information $h_t^0$. Hence, we have 

\vspace{-2mm}
\begin{small}
\begin{align}
&\ee^{\phi_{0:t-1}^{prs}} \big[[\tilde \phi_t^{1}(H_t^0)](X_t^1,M_t^1)  \vert h_t^0,m_t^1\big] \notag \\
&= \ee^{\phi_{0:t-1}^{prs}} \big[[\tilde \phi_t^{1}(h_t^0)](X_t^1,m_t^1)  \vert h_t^0\big] 
=  \int [\tilde \phi_t^{1}(h_t^0)](x_t^1,m_t^1) \theta_t(dx_t^1).
\end{align}
\end{small}
Now, if we define $\bar  q_t = \bar \phi_t^1(h_t^0)$, $\tilde q_t = \tilde \phi_t^1(h_t^0)$, and 
\begin{small}
\begin{align}
\bar{\mathcal{Q}}&=
\{ \bar q_t: \mathcal{M}^1 \to \R^{d_U^1} \text{ measurable}\}, \notag \\
\mathcal{\tilde Q}(\theta_t) &=
\{ \tilde q_t: \R^{d_X^1} \times \mathcal{M}^1 \to \R^{d_U^1} \text{ measurable}, \notag \\
&\hspace{4mm} \int \tilde q_t(x_t^1,m_t^1) \theta_t(dx_t^1) =0 \quad \text{ for all $m_t^1 \in \mathcal{M}^1$}
\},
\label{Q_set}
\end{align}
\end{small}
then we have  $\bar q_t \in \bar{\mathcal{Q}}$ and $\tilde q_t \in \tilde{\mathcal{Q}}(\theta_t)$. This together with the representation $U_t^1$ of \eqref{decomposition} suggests that at each time $t$, for any $h^0_t\in\mathcal{H}^0_t$ and its corresponding common beliefs $\theta_t$, finding a function $\rho_t \in \mathcal{P}$ is equivalent to finding two functions $\bar q_t \in \bar{\mathcal{Q}}$ and $\tilde q_t \in \tilde{\mathcal{Q}}(\theta_t)$. The following lemma states this point formally.
\begin{lemma}
\label{decomposition_equivalence}
In Problem \ref{problem:equivalent}, for any $h^0_t\in\mathcal{H}^0_t$, let $\theta_t$ be its corresponding common belief defined by \eqref{eq:thetat_margin}-\eqref{update_theta}. Then,
\begin{align}
\mathcal{P} = \{\bar q_t \circ h_2 + \tilde q_t:  \bar q_t \in \bar{\mathcal{Q}}, \tilde q_t \in \tilde{\mathcal{Q}}(\theta_t)  \}
\label{set_equivalence}
\end{align}
where $h_2: \R^{d_X^1} \times \mathcal{M}^1 \to  \mathcal{M}^1$ is a projection function defined as $h_2(x^1,m^1) = m^1$. 
\end{lemma}

\begin{proof}
See Appendix \ref{Proof_lemma_decomposition} for a proof.
\end{proof}

According to Lemma \ref{decomposition_equivalence}, any prescription strategy $\phi^{prs}= (\phi^{0}, \phi^{1}) \in \Phi^{prs}$ can be equivalently represented by $(\phi^{0}, \bar \phi^{1}, \tilde \phi^{1})$. In the new representation, the control action $u_t^1$ of $C^1$ applied to the system described by \eqref{Model:system_0}-\eqref{Model:channel} is 
$u_t^{1} = \bar q_t(M_t^1) + \tilde q_t(X_t^1,M_t^1)$ where $\bar  q_t = \bar \phi_t^1(h_t^0)$ and $\tilde q_t = \tilde \phi_t^1(h_t^0)$. In the following, we use the equivalent representation $(\phi^{0}, \bar \phi^{1}, \tilde \phi^{1})$ for any prescription strategy $\phi^{prs} \in \Phi^{prs}$.



\subsection{Modified Dynamic Program for Problem \ref{problem:equivalent}}

The following theorem provides a modified dynamic program for optimal prescription strategies of Problem \ref{problem:equivalent} based on the new representation of strategies described above.

\begin{theorem}
\label{thm:structure}
Theorem \ref{thm:general_structure} holds if $u^{prs}_t=(u^{0}_t, \bar q_t \circ h_2 + \tilde q_t)$ and the minimization in \eqref{eq:DP_V} is over $u^0_t \in \R^{d_U^0}, \bar q_t \in \bar{\mathcal{Q}}, \tilde q_t \in \tilde{\mathcal{Q}}(\theta_t) $.
\end{theorem}
\begin{proof}
The proof is obtained by applying Lemma \ref{decomposition_equivalence} to Theorem  \ref{thm:general_structure}.
\end{proof} 
Note that although each step of the DP of Theorem \ref{thm:structure} involves a functional optimization similar to Theorem \ref{thm:general_structure}, in the next section, we show that the functions satisfying \eqref{eq:DP_V} exist and 
it is possible to find a solution to the DP of Theorem \ref{thm:structure}. This solution can then be used to provide optimal strategies for the controllers.

\begin{remark}
The decomposition of the controller $C^1$'s strategies proposed in \eqref{decomposition}
is different from the one in \cite{Ouyang_Asghari_Nayyar:2016,Asghari2016optimal} where $U_t^1$ is decomposed into two terms: the conditional mean of $U_t^1$ given the common information $H_t^0$ and the deviation of $U_t^1$ from the mean. One can follow the decomposition of \cite{Ouyang_Asghari_Nayyar:2016,Asghari2016optimal} here and observe that it fails to provide any structure for solving the DP of Theorem \ref{thm:structure}. 
\end{remark}

\section{Optimal Control Strategies}
\label{sec:solution}
In this section, we identify the structure of the value function in the modified  dynamic program described in Section \ref{sec:decomposition_DP}.
Using the structure, we explicitly solve the dynamic program and obtain the optimal strategies for Problem \ref{problem1}.

\begin{theorem}
\label{thm:Sol_mode}
For any $x_t^0 \in \R^{d_X^0}$, $m_t^0 \in \mathcal{M}^0$, $\theta_t \in \Delta(\R^{d_X^1})$, and $\tilde z_t \in \tilde{\mathcal{M}}^1$ at time $t$, {\color{black}there} exist positive semi-definite matrices $P_t(m_t^0,\tilde z_t)$ and $\tilde{P}_t(m_t^0,\tilde z_t)$ such that the value function of the dynamic program \eqref{eq:DP_V} is given by  

\vspace{-4mm}
\begin{small}
\begin{align}
V_t(x_{t}^0, m_{t}^0, \theta_{t}, \tilde z_{t}) &= 
QF \Big(P_t(m_t^0,\tilde z_t), \vecc \big(x^0_t,\mu(\theta_t) \big) \Big) \notag \\
& + \tr \Big( \tilde{P}_t(m_t^0,\tilde z_t) \cov(\theta_t) \Big) + e_t,
\label{eq:Vt_mode}
\end{align}
\end{small}
where $e_t$ is a function of statistics of $W_{t+1:T}^{0:1}$ and matrices $P_{t+1:T}, \tilde P_{t+1:T}$.

Furthermore, for any $x_t^0 \in \R^{d_X^0}$, $m_t^0 \in \mathcal{M}^0$, $\theta_t \in \Delta(\R^{d_X^1})$, and $\tilde z_t \in \tilde{\mathcal{M}}^1$ at time $t$, there exist matrices $K_t(m_t^0,\tilde z_t)$ and $\tilde K_t(m_t^0,\tilde z_t)$ such that 
the minimizing $u_t^0, \bar q_t, \tilde q_t$ are
\begin{itemize}
\item If $\tilde z_t = \emptyset$:
\begin{small}
\begin{align}
&\begin{bmatrix}
u^{0*}_t \\
\bar q_t^*(1) \\
\vdots \\
\bar q_t^*(\kappa^1)
\end{bmatrix}
=K_t(m_t^0, \emptyset)
\begin{bmatrix}
x_t^0 \\
\mu(\theta_t)
\end{bmatrix},
\label{eq:opt_ubar_mode_empty} \\
& \hspace{-4mm} \tilde q^{*}_t(x_t^1,m_t^1) = \tilde K_t(m_t^0,m_t^1)
\big(x_t^1 - \mu(\theta_t)  \big)
\notag \\
& \hspace{30mm} \forall x_t^1 \in \R^{d_X^1}, \forall m_t^1 \in \mathcal{M}^1.
\label{eq:opt_gamma_mode_empty}
\end{align}
\end{small}
\item If $\tilde z_t = l$ for any $l \in \mathcal{M}^1$:

\begin{small}
\begin{align}
&\begin{bmatrix}
u^{0*}_t \\
\bar q_t^*(l)
\end{bmatrix}
=K_t(m_t^0, l)
\begin{bmatrix}
x_t^0 \\
\mu(\theta_t)
\end{bmatrix},
\label{eq:opt_ubar_mode} \\
& \bar q_t^{*}(m_t^1) = 0, \quad \forall m_t^1 \in \mathcal{M}^1 \setminus \{ l \}, 
\label{eq:opt_qbar_mode_l}
\\
&\tilde q^{*}_t(x_t^1,m_t^1) = 0, \quad  \forall x_t^1 \in \R^{d_X^1}, \forall m_t^1 \in \mathcal{M}^1.
\label{eq:opt_gamma_mode}
\end{align}
\end{small}
\end{itemize}
\end{theorem}
\begin{proof}
See Appendix \ref{Proof_theorem_Sol_mode} for a proof.
\end{proof}
The matrices $P_t$ and $\tilde P_t$ can be explicitly computed using coupled "Riccati-like" backward recursions, and the gain matrices $K_t$ and $\tilde K_t$ are simple functions of $P_t$ and $\tilde P_t$. The computation of these matrices can be done offline with computational complexity similar to that of an optimal centralized LQR controller. See Appendix \ref{Recursions} for the recursions for $P_t$ and $\tilde P_t$ and the equations for $K_t$ and $\tilde K_t$

%
From Theorem \ref{thm:Sol_mode}, we can explicitly compute the optimal strategies for Problem \ref{problem1}. The optimal strategies of controllers $C^{0}$ and $C^{1}$ are given in the following theorem.

\begin{theorem}
\label{thm:opt_strategies}
For any $x_t^0 \in \R^{d_X^0}$ and $m_t^0 \in \mathcal{M}^0$, the optimal strategies of Problem \ref{problem1} are given by
$u^{0*}_t$ and $u^{1*}_t = \rho_t^{*}(x_t^1,m_t^1)  = \bar q_t^{*} (m_t^1 ) + \tilde q_t^{*} (x_t^1, m_t^1)$ where (i) $u^{0*}_t,  \bar q_t^{*},  \tilde q_t^{*}$ are described by \eqref{eq:opt_ubar_mode_empty}-\eqref{eq:opt_gamma_mode} and (ii)
$\mu(\theta_t) = \hat x_t^1$ is the estimate (conditional expectation) of $X_t^1$ based on the common information $h^0_t$ and it can be computed recursively according to

\vspace{-4mm}
\begin{small}
\begin{align}
&\hat x_0^1=\left\{
\begin{array}{ll}
\mu(\pi_{X_0^1}) & \text{ if }z_{0}= \emptyset,\\
 x_{0}^1 & \text{ if }z_{0} = x_{0}^1.
\end{array}\right.
\label{eq:estimator_0}
\end{align}
\end{small}
\vspace{-2mm}
\begin{small}
\begin{align}
&\hat x_{t+1}^1 =
\notag\\
&\left\{
\begin{array}{ll}
\hspace{-3mm}
\sum\limits_{m_t^1 \in \mathcal{M}^1} \pi_{M^1}(m_t^1)
[D(m_t^{0:1})]_{1\bullet} \begin{bmatrix} x_t^0 \\ \hat x_t^1\\ \bar u^{0*}_t \\ \bar q_t^{*}(m_t^1) \end{bmatrix} 
  &\hspace{-0mm} \text{if }z_{t+1}= \tilde z_t =\emptyset, \\
\ [D(m_t^{0:1})]_{1\bullet}  \begin{bmatrix} x_t^0 \\ \hat x_t^1\\ \bar u^{0*}_t \\ \bar q_t^{*}(m_t^1) \end{bmatrix} 
&  \hspace{-4mm} \text{if }z_{t+1}= \emptyset, \tilde z_t = m_t^1, \\
 \hspace{0mm} x_{t+1}^1 & \hspace{-4mm}  \text{if }z_{t+1} = x_{t+1}^1,\\
\end{array}\right.
\label{eq:estimator_t}
\end{align}
\end{small}
\end{theorem}

\begin{proof}
See Appendix \ref{Proof_Thm_optimal_strategies} for a proof.
\end{proof}

\begin{remark}
When the transmission from the controller $C^1$ to the controller $C^0$ is successful ($\gamma_t=1$), the controller $C^0$ is aware of the state $x_t^1$ and local mode $m_t^1$. In this case, the term inside the minimization of \eqref{eq:DP_V} does not depend on function $\tilde q_t$ and also it does not depends on $\bar q_t(\ell)$ for all $\ell \neq m_t^1$. Hence, they can be chosen arbitrarily or set to be zero as described in \eqref{eq:opt_qbar_mode_l} and \eqref{eq:opt_gamma_mode}.
\end{remark}


\begin{remark}
If there is only possible value for the local mode, that is, $|\mathcal{M}^1|=1$, the controller $C^0$ knows this mode irrespective of the state of the link from the controller $C^1$ to the controller $C^0$. In this case, the optimal controller for this problem is described by \eqref{eq:opt_ubar_mode_empty} and \eqref{eq:opt_gamma_mode_empty} if $\mathcal{M}^1 = \{1 \}$.
\end{remark}

\section{Conclusion}
\label{sec:conclusion}

{\color{black}
We considered a discrete-time stochastically switched system consisting of two plants, global plant and local plant, and two controllers, global controller and local controller. 
We assumed the presence of an unreliable TCP-like communication channel through which the local controller can inform the global controller of the local plant's state and mode. We obtained explicit optimal strategies for the two controllers. In the optimal strategies, both controllers compute a common estimate of the local plant's state. The global controller's action is linear in the state of the global plant and the common estimated state, and the local controller's action is linear in the actual states of both plants and the common estimated state. 
Furthermore, the gain matrices for the global controller depend on the global mode and its observation about the local mode, while the gain matrices for the local controller depend on the actual modes of both plants and the global controller's observation about the local mode.
}


\bibliographystyle{ieeetr}
\bibliography{IEEEabrv,References,packet_drop,collection,switched_systems}

\appendices
\section{Short-hands and Operators}
We first define some shorthands for the simplicity of the presentation. For a vector $x$ and a matrix $G$, we define 
$QF(G,x)= x^\tp G\, x = \tr(Gxx^\tp)$. Furthermore, for a matrix $G = \begin{bmatrix}
G^{11} & G^{12} \\ G^{21} & G^{22}
\end{bmatrix}$, we define $\textit{SC}(G,G^{22}):=G^{11} - G^{12}(G^{22})^{-1}G^{21}$ as the Schur complement of $G^{22}$ in $G$. For matrices $G^1,\ldots, G^k$, we define $HC(G^1,\ldots, G^k): = [G^1,\ldots, G^k]$ to represent the horizontal concatenation of $G^1,\ldots, G^k$. Furthermore, we use $\mathds{1}_{E}(\cdot)$ to denote the indicator function of set $E$, that is, $\mathds{1}_{E}(x) = 1$ if $x \in E$, and $0$ otherwise.

Operators:
\begin{itemize}
\item Let $G$ be a collection of matrices $G(m^0, \tilde z) \in \R^{d_X\times d_X}$ for all $m^0 \in \mathcal{M}^0$ and $\tilde z \in \mathcal{\tilde M}^1 = \mathcal{M}^1 \cup \{\emptyset \}$. Then, define
\begin{small}
\begin{align}
\Pi(G) = \ee[G(M_t^0, \tilde Z_t)] = \sum_{\gamma \in \{0,1\}}p(\gamma) \Pi(G,\gamma).
\label{operator_Pi}
\end{align}
\end{small}
where for any $\gamma \in \{0,1\}$
\begin{small}
\begin{align}
\Pi(G, \gamma) = \ee[G(M_t^0, \tilde Z_t) \vert \Gamma_{t} = \gamma].
\label{operator_Pi_gamma}
\end{align}
\end{small}
Note that
\begin{small}
\begin{align}
\Pi(G,0) 
&= 
\sum_{\substack{m^0 \in \mathcal{M}^0}}
G(m^0,\emptyset)\pi_{M^0}(m^0),  \notag \\
\Pi(G,1) 
& = 
\sum_{\substack{m^0 \in \mathcal{M}^0}} \sum_{\substack{m^1 \in \mathcal{M}^1}}
G(m^{0:1})\pi_{M^0}(m^0)\pi_{M^1}(m^1).
\notag
\end{align}
\end{small}
\item For $n = 1,2$, let $G^n$ be a collection of matrices $G^n(m^0, \tilde z) \in \R^{d_X\times d_X}$ for all $m^0 \in \mathcal{M}^0$ and $\tilde z \in \mathcal{\tilde M}^1$. Then, define
\begin{small}
\begin{align}
\Psi(G^1, G^2) = p(0)\Pi(G^1,0) +  p(1) \Pi(G^2,1).
\label{operator_Psi}
\end{align}
\end{small}
\end{itemize}
Matrices:
\begin{itemize}
\item For each $m^1 \in \mathcal{M}^1$, $L_{m^1} \in \R^{(d_X + d_U^0 + d_U^1) \times (d_X + d_U^0 + \kappa^1 d_U^1)}$ given by 
\begin{small}
\begin{align}
\hspace{-10mm}
L_{m^1} &= HC \Big(
\begin{bmatrix}
\mathbf{I}_{d_X}  &\mathbf{0}   \\
 \mathbf{0} &\mathbf{I}_{d_U^0} \\
 \mathbf{0} & \mathbf{0}
\end{bmatrix}
,\mathbf{0}_{(d_X + d_U^0 + d_U^1)  \times (d_X + d_U^0 + (m^1-1) d_U^1)}, \notag \\
&\begin{bmatrix}
\mathbf{0} \\
\mathbf{0}  \\
 \mathbf{I}_{d_U^1}
\end{bmatrix}
,
\mathbf{0}_{(d_X + d_U^0 + d_U^1)  \times (d_X + d_U^0 + (\kappa^1 - m^1) d_U^1)}
\Big).
\label{Lm_matrix}
\end{align}
\end{small}
\begin{small}
\item $Q_t(M_t^{0:1}) = \begin{bmatrix}
Q^{00}_t(M_t^{0:1}) & Q^{01}_t(M_t^{0:1}) \\
Q^{10}_t(M_t^{0:1}) & Q^{11}_t(M_t^{0:1}) 
\end{bmatrix}$, 
\vspace{1mm}
\item $R_t(M_t^{0:1}) = \begin{bmatrix}
R^{00}_t(M_t^{0:1}) & R^{01}_t(M_t^{0:1}) \\
R^{10}_t(M_t^{0:1}) & R^{11}_t(M_t^{0:1}) 
\end{bmatrix}$, 
\vspace{1mm}
\item $D^{11}(m^{0:1}) = \begin{bmatrix}
 A^{11}(m^{0:1}) & B^{11}(m^{0:1})
\end{bmatrix}$,
\vspace{1mm}
\item $D^{\text{aug}}(m^{0:1}) = \begin{bmatrix}
A(m^{0:1}) & B(m^{0:1})
\end{bmatrix}L_{m^1}$,
\vspace{1mm}
\item $D^{\emptyset}(m^{0}) = \sum_{m^1 \in \mathcal{M}^1} D^{\text{aug}}(m^{0:1}) \pi_{M^1}(m^1)$,
\vspace{1mm}
\item $C_t(m^{0:1})= \begin{bmatrix}
Q_t(m^{0:1}) & \mathbf{0} \\
\mathbf{0} & R_t(m^{0:1})
\end{bmatrix}$,
\vspace{1mm}
\item $C_t^{\emptyset}(m^{0}) = 
\sum_{m^1 \in \mathcal{M}^1} L_{m^1}^{\tp}C_t(m^{0:1})L_{m^1} \pi_{M^1}(m^1)$, 
\vspace{1mm}
\item $C_t^{11}(m^{0:1})= \begin{bmatrix}
Q_t^{11}(m^{0:1}) & \mathbf{0} \\
\mathbf{0} & R_t^{11}(m^{0:1})
\end{bmatrix}$.
\end{small}
\end{itemize}

\section{Recursions for $P_t,\tilde P_t$ and equations for $K_t,\tilde K_t$}
\label{Recursions}
For any $m^0 \in \mathcal{M}^0$, $m^1 \in \mathcal{M}^1$, and $\tilde z \in \mathcal{\tilde{M}}^1$, matrices $P_t(m^0,\tilde z)$ and $\tilde P_t(m^0,\tilde z)$ are recursively calculated as follows,
\begin{small}
\begin{align}
&P_{T+1}(m^0,\tilde z) = \begin{bmatrix}
P_{T+1}^{00}(m^0,\tilde z) & P_{T+1}^{01}(m^0,\tilde z) \\
P_{T+1}^{10}(m^0,\tilde z)  & P_{T+1}^{11}(m^0,\tilde z)
\end{bmatrix} = \mathbf{0}, 
\label{eq:recursion_P_T} \\
& \tilde P_{T+1}(m^0,\tilde z)  = \mathbf{0},
\label{eq:recursion_tilde_P_T}
\end{align}
\end{small}
\begin{small}
\begin{align}
&E_t^{\emptyset}(m^0) = D^{\emptyset}(m^0)^{\tp} \Pi(P_{t+1})D^{\emptyset}(m^0),
\label{eq:recursion_E_t} 
\\
&F_t^{\emptyset}(m^0) = \notag \\
&\sum_{\substack{m^1 \in \mathcal{M}^1}}
[D^{\text{aug}}(m^{0:1})]_{1\bullet}^\tp \Psi(\tilde P_{t+1}, P_{t+1}^{11}) [D^{\text{aug}}(m^{0:1})]_{1\bullet} \pi_{M^1}(m^1)
\notag \\
& \quad - [D^{\emptyset}(m^{0})]_{1\bullet}^{\tp} \Psi(\tilde P_{t+1}, P_{t+1}^{11}) [D^{\emptyset}(m^{0}) ]_{1\bullet}
\label{eq:recursion_F_t}
\end{align}
\end{small}
\begin{small}
\begin{align}
& H_t(m^0,\tilde z)  =  \begin{bmatrix}
H_t^{XX}(m^0,\tilde z)  & H_t^{XU}(m^0,\tilde z)  \\
H_t^{UX}(m^0,\tilde z)  & H_t^{UU}(m^0,\tilde z) 
\end{bmatrix}
\notag \\
&=\left\{
\begin{array}{ll}
C_t^{\emptyset}(m^0) + E_t^{\emptyset}(m^0) +F_t^{\emptyset}(m^0) & \text{ if } \tilde z= \emptyset,\\
C_t(m^{0:1}) + D(m^{0:1})^{\tp} \Pi(P_{t+1}) D(m^{0:1}) & \text{ if } \tilde z = m^1.
\end{array}\right.
\label{eq:recursion_H_t} 
\\
& P_t(m^0,\tilde z)  = \textit{SC}\big(H_t(m^0,\tilde z), H_t^{UU}(m^0,\tilde z) \big) 
\label{eq:recursion_P_t}
\end{align}
\end{small}
\begin{small}
\begin{align}
&\tilde H_t(m^{0:1}) = 
\begin{bmatrix}
\tilde H_t^{X^1X^1} (m^{0:1}) & \tilde H_t^{X^1U^1} (m^{0:1}) \\
\tilde H_t^{U^1X^1} (m^{0:1}) & \tilde H_t^{U^1U^1} (m^{0:1})
\end{bmatrix} \notag \\
&=C_t^{11} (m^{0:1}) +D^{11}(m^{0:1})^{\tp} \Psi(\tilde P_{t+1}, P_{t+1}^{11}) D^{11}(m^{0:1}),
\label{eq:recursion_tilde_H_t}
\end{align}
\end{small}
\begin{small}
\begin{align}
& \tilde P_t(m^0,\tilde z)  = \notag \\
&\left\{
\begin{array}{ll}
 \sum_{m^1 \in \mathcal{M}^1} \pi_{M^1}(m^1)
\textit{SC}\big(\tilde H_t(m^{0:1}), \tilde H_t^{U^1U^1}(m^{0:1}) \big)  & \text{if } \tilde z= \emptyset,\\
\textit{SC}\big(\tilde H_t(m^{0:1}), \tilde H_t^{U^1U^1}(m^{0:1}) \big)  & \text{if } \tilde z = m^1.
\end{array}\right.
\label{eq:recursion_P_t} 
\end{align}
\end{small}
Furthermore, at each time $t$, the gain matrices $K_t(m^0,\tilde z)$ and $\tilde K_t(m^{0:1})$ for any $m^0 \in \mathcal{M}^0$, $m^1 \in \mathcal{M}^1$, and $\tilde z \in \mathcal{\tilde{M}}^1$ are calculated as follows,
\begin{small}
\begin{align}
&K_t(m^{0},\tilde z) = -\big(H_t^{UU}(m^{0},\tilde z)\big)^{-1} H_t^{UX}(m^{0},\tilde z),
\label{K_t} \\
& \tilde K_t(m^{0:1})= 
- \big(\tilde H_t^{U^1U^1}(m_t^{0:1})\big)^{-1} \tilde H_t^{U^1X^1}(m_t^{0:1}).
\label{tilde K_t}
\end{align}
\end{small}

\section{Preliminary Results}

\newtheorem{claim}{Claim}

\begin{claim}
\label{clm:independence_state_mode}
Consider a feasible prescription strategy $\phi^{\diamond}_t = \phi^{prs}_{0:t-1} \in \Phi^{prs}$. Then, the random vectors $M_t^0$,$X_{t}^{0}$ and $X_{t}^1$ are conditionally independent given the common information $H^0_{t-1}$. 
That is, for any measurable sets $E^0 \subset \R^{d_X^0}$, $E^1 \subset \R^{d_X^1}$, and $F \subset \mathcal{M}^1$,

\begin{small}
\begin{align}
&\prob^{\phi^{\diamond}_t}(M_t^0 \in F, X_{t}^0 \in E^0, X_{t}^1 \in E^1 | H^0_{t-1}) \notag \\
&= \prob(M_t^0 \in F) \prob^{\phi^{\diamond}_t}(X_{t}^0 \in E^0| H^0_{t-1})
\prob^{\phi^{\diamond}_t}(X_{t}^1 \in E^1 | H^0_{t-1}).
\end{align}
\end{small}
\end{claim}

\begin{proof}
Note that given any realization $h_{t-1}^0$ of $H_{t-1}^0$, $X_t^0$ is a function of only $W_{t-1}^0$ while $X_t^1$ is a function of $X_0^1, W_{0:t-1}^1,M_{0:t-1}^1$. Then, the proof is easily resulted form the fact that random variables in the collection $\{M_{t}^0, X_0^1, W_{0:t-1}^1,M_{0:t-1}^1, W_{t-1}^0 \} $ are independent.
\end{proof}

\begin{lemma}
\label{lm:omega_equation}
Consider a feasible prescription strategy $\phi^{\diamond}_t = \phi^{prs}_{0:t-1} \in \Phi^{prs}$. Then, for any $h_t^0 \in \mathcal{H}_t^0$ and for any $m \in \mathcal{M}^1$,  the $C^0$'s belief $\prob^{\phi^{\diamond}_t}(M_t^1 =m | h^0_{t})$ on the local mode $M_t^1$ is $\omega^{\tilde z_t}(m)$ where
\begin{itemize}
\item  If $\tilde z_{t} = m_{t}^1$, then  $\omega^{m_t^1}(m)= \mathds{1}_{\{m\}}(m_{t}^1)$.
\item If $\tilde z_{t} = \emptyset$, then $\omega^{\emptyset}(m) = \pi_{M_t^1}(m)$.
\end{itemize}
\end{lemma}

\begin{proof}
Note that $\tilde z_t$ is included in $h_t^0$. Hence, when $\tilde z_{t} = m_{t}^1$, we have
\begin{small}
\begin{align}
\prob^{\phi^{\diamond}_t}(M_t^1 =m | h^0_{t}) = \prob(m_t^1 =m) = \mathds{1}_{\{m\}}(m_{t}^1).
\end{align}
\end{small}
When $\tilde z_{t} = \emptyset$, we have
\begin{small}
\begin{align}
&\prob^{\phi^{\diamond}_t}(M_t^1 =m | h^0_{t}) =
\prob^{\phi^{\diamond}_t}(M_t^1 =m | h^0_{t-1}, x_t^0,m_t^0,u_{t-1}^0) \notag \\
&= \prob^{\phi^{\diamond}_t}(M_t^1 =m) = \pi_{M_t^1}(m)
\label{belief_mode_proof}
\end{align}
\end{small}
Note that $X_{t}^0 = [D(m_{t-1}^{0:1})]_{0\bullet} S_{t-1} + W_{t-1}^0$ where $S_{t-1} = \vecc(X_{t-1}^{0:1},U_{t-1}^{0:1})$. Hence, the second equality of \eqref{belief_mode_proof} is true because 1) the collection $\{X_{t}^0, H_{t-1}^0, U_{t-1}^0 \}$ depend on random variables from time $0$ to $t-1$ and $M_t^1$ is independent of all previous random variables; 2) $M_t^1$ is independent of $M_t^0$. 
\end{proof}

\section{Conditional independence of $X_t^{1}$ and $M_t^1$ given $H^0_t$}
\label{Proof:independence_of_state_mode}
\begin{lemma}
\label{lm:independence_of_state_mode}
Consider any feasible prescription strategy $\phi^{\diamond}_t = \phi^{prs}_{0:t-1} \in \Phi^{prs}$. Then, $X_t^{1}$ and $M_t^1$ are conditionally independent given the common information $H^0_t$. 
\end{lemma}
\begin{proof}
According to \eqref{Model:info}, for any measurable sets $E \subset \R^{d_X^1}$  and $F \subset \mathcal{M}^1$, we have,
\begin{small}
\begin{align}
& \prob^{\phi^{\diamond}_t}(X_t^1 \in E, M_t^1 \in F | H^0_{t}) \notag \\
&= \prob^{\phi^{\diamond}_t}(X_t^1 \in E, M_t^1 \in F | H^0_{t-1}, X_t^0,M_t^0, Z_t, \tilde Z_t, U_{t-1}^0)
\label{eq:proof_indep}
\end{align}
\end{small}
Now, consider two following cases. If $\Gamma_t^0 = 1$, we have
\begin{small}
\begin{align}
&\prob^{\phi^{\diamond}_t}(X_t^1 \in E, M_t^1 \in F | H^0_{t-1}, X_t^0,M_t^0, Z_t, \tilde Z_t, U_{t-1}^0)
\ind_{\{\Gamma_t =1\}} \notag \\ 
&=\prob^{\phi^{\diamond}_t}(X_t^1 \in E, M_t^1 \in F | H^0_{t-1}, X_t^0,M_t^0, X_t^1, M_t^1, U_{t-1}^0)
\ind_{\{\Gamma_t =1\}} \notag \\ 
&=\prob^{\phi^{\diamond}_t}(X_t^1 \in E | H^0_{t})
\prob^{\phi^{\diamond}_t}(M_t^1 \in F | H^0_{t})
\ind_{\{\Gamma_t =1\}}.
\end{align}
\end{small}
 If $\Gamma_t^0 = 0$, we have
\begin{small}
\begin{align}
&\prob^{\phi^{\diamond}_t}(X_t^1 \in E, M_t^1 \in F | H^0_{t-1}, X_t^0,M_t^0, Z_t, \tilde Z_t, U_{t-1}^0)
\ind_{\{\Gamma_t =0\}} \notag \\ 
&=\prob^{\phi^{\diamond}_t}(X_t^1 \in E, M_t^1 \in F | H^0_{t-1}, X_t^0,M_t^0, U_{t-1}^0)
\ind_{\{\Gamma_t =0\}} \notag \\ 
&=\prob^{\phi^{\diamond}_t}(X_t^1 \in E| H^0_{t})
\prob^{\phi^{\diamond}_t}(M_t^1 \in F | H^0_{t})
\ind_{\{\Gamma_t =0\}}.
\label{proof_indep_2}
\end{align}
\end{small}
Note that $X_{t}^0 = [D(M_{t-1}^{0:1})]_{0\bullet} S_{t-1} + W_{t-1}^0$, $X_{t}^1 = [D(M_{t-1}^{0:1})]_{1 \bullet} S_{t-1} + W_{t-1}^1$ where $S_{t-1} = \vecc(X_{t-1}^{0:1},U_{t-1}^{0:1})$. Hence, the penultimate equality of \eqref{proof_indep_2} is true because 1) the collection $\{X_{t}^0, X_{t}^1, H_{t-1}^0, U_{t-1}^0 \}$ depend on random variables from time $0$ to $t-1$ and $M_t^1$ is independent of all previous random variables; 2) $M_t^1$ is independent of $M_t^0$. This completes the proof of Lemma \ref{lm:independence_of_state_mode}.
\end{proof}
\section{Recursive Update for Common belief $\theta_t$}
\label{Proof_Lm_belief_update}

\begin{lemma}
\label{lm:beliefupdate}

For any feasible prescription strategy $\phi^{prs} \in \Phi^{prs}$ and for any $h_t^0 \in \mathcal{H}_t^0$, the common belief $\theta_t$ can be updated according to 
\begin{align}
\theta_{t+1} = \psi_t(\theta_t, u_t^{prs}, x_t^0, m_t^0, \tilde z_t, z_{t+1}),
\label{update_theta}
\end{align}
where $u_t^{prs} = (u_t^0, \rho_t)= \phi^{prs}_{t}(h^0_t)$
and for any measurable set $E \subset  \R^{d_X^1}$, 
\begin{small}
\begin{align}
\theta_0 (E)  
&= \left\{\begin{array}{ll}
\pi_{X_0^1}(E)
& \text{ if }z_{0}=\emptyset, 
\\
\mathds{1}_{E}(x^1_0) & \text{ if }z_{0}=x_{0}^1.
\end{array}\right.
\label{eq:theta0}
\end{align}
\end{small}
\begin{itemize}
\item If $z_{t+1} = x_{t+1}^1$, then 
\begin{small}
\begin{align}
[\psi_t(\theta_t,u^{prs}_t, x_t^0, m_t^0, \tilde z_t, x_{t+1}^1)](E)= \mathds{1}_{E}(x^1_{t+1}).
\label{belief_theta_gamma_1}
\end{align}
\end{small}
\item If $z_{t+1} = \emptyset$, then 
\begin{small}
\begin{align}
&\hspace{-3mm}
 [\psi_t(\theta_t,u^{prs}_t, x_t^0,m_t^0, \tilde z_t, \emptyset)] (E)=  
\notag \\
& \hspace{-5mm} \sum_{m_t^1 \in \mathcal{M}^1} \int \int  \mathds{1}_{E}
\Big( [D(m_t^{0:1})]_{1\bullet} \vecc \big(x_t^{0:1}, u_t^0, \rho_t(x_t^1,m_t^1) \big) + w_t^1
\Big)
  \notag \\
&\hspace{17mm} \times \theta_t(dx_t^1) \omega^{\tilde z_{t}}(m_t^1)\pi_{W_t^1}(dw_t^1).
\label{eq:psit}
\end{align}
\end{small}
\end{itemize}

\end{lemma}

\begin{proof}
To prove Lemma \ref{lm:beliefupdate} we proceed as follows. For any feasible prescription strategy $\phi^{prs} \in \Phi^{prs}$ and for any $h_t^0 \in \mathcal{H}_t^0$, we recursively define $\nu_t (h_t^0) \in \Delta(\R^{d_X^1})$ as follows:

For any measurable set $E \subset \R^{d_X^1}$,
\begin{small}
\begin{align}
[\nu_0(h_0^0)] (E)
&= \left\{\begin{array}{ll}
\pi_{X_0^1}(E)
& \text{ if }z_{0}=\emptyset, 
\\
\mathds{1}_{E}(x^1_0) & \text{ if }z_{0}=x_{0}^1.
\end{array}\right.
\label{eq:theta0}
\end{align}
\end{small}
For any measurable set $E \subset \R^{d_X^1}$,
\begin{small}
\begin{align}
[\nu_{t+1}(h_{t+1}^0)] (E)
&= [\psi_t(\nu_t(h_t^0),u^{prs}_t, x_t^0, m_t^0, \tilde z_t, z_{t+1})](E) ,
\label{eq:theta_update}
\end{align}
\end{small}
where $u_t^{prs} = (u_t^0, \rho_t)= \phi^{prs}_{t}(h^0_t)$ and $\psi_t^n(\nu_t(h_t^0),u^{prs}_t, x_t^0, m_t^0, \tilde z_t, z_{t+1})$ is defined as follows:
\begin{itemize}
\item If $z_{t+1} = x_{t+1}^1$, then 
\begin{align}
[\psi_t(\nu_t(h_t^0),u^{prs}_t, x_t^0, m_t^0, \tilde z_t, x_{t+1}^1)](E)= \mathds{1}_{E}(x^1_{t+1}).
\label{belief_theta_gamma_1}
\end{align}
\item If $z_{t+1} = \emptyset$, then 
\begin{small}
\begin{align}
& [\psi_t(\nu_t(h_t^0),u^{prs}_t, x_t^0,m_t^0, \tilde z_t, \emptyset)] (E)=  
\notag \\
& \hspace{-11mm}\sum_{m_t^1 \in \mathcal{M}^1} \int \int  \mathds{1}_{E}
\Big( [D(m_t^{0:1})]_{1 \bullet} \vecc \big(x_t^{0:1}, u_t^0, \rho_t(x_t^1,m_t^1) \big) + w_t^1
\Big) \notag \\
&\hspace{10mm} \times [\nu_t(h_t^0)](dx_t^1) \omega^{\tilde z_{t}}(m_t^1) \pi_{W_t^1}(dw_t^1).
\label{eq:psit}
\end{align}
\end{small}
\end{itemize}
Then, we show that $\nu_t$ is a conditional probability of $X_t^1$ given $H^0_t$, that is $[\nu_{t}(H_{t}^0)] (E) = \prob^{\phi^{\diamond}_t}(X_t^1 \in E | H^0_{t})$ where $\phi^{\diamond}_t = \phi^{prs}_{0:t-1}$.

First, note that from \eqref{eq:theta0}-\eqref{eq:psit}, $[\nu_t(\cdot)](E): \mathcal{H}_t^0 \mapsto \R$ is a measurable function. To show that
$[\nu_{t}(H_{t}^0)] (E) = \prob^{\phi^{\diamond}_t}(X_t^1 \in E | H^0_{t})$,
first note that for any $t$

\begin{small}
\begin{align}
&\prob^{\phi^{\diamond}_t}(X_{t}^1 \in E | H^0_{t}) = \prob^{\phi^{\diamond}_t}(X_{t}^1 \in E | H^0_{t-1}, X_t^0, M_t^0, Z_{t}, \tilde Z_t) \notag \\
&= \prob^{\phi^{\diamond}_t}(X_{t}^1 \in E | H^0_{t-1}, Z_{t}^n, \tilde Z_t, ),
\label{eq:nu_t_proof}
\end{align}
\end{small}
where the last equality is true due to Claim \ref{clm:independence_state_mode}.

We now prove by induction that 
\begin{align}
[\nu_{t}(H_{t}^0)] (E) = \prob^{\phi^{\diamond}_t}(X_{t}^1 \in E | H^0_{t-1}, Z_{t}, \tilde Z_t).
\label{nu_claim}
\end{align}
At time $t=0$, since $\Gamma_0 \in \{0,1\}$, consider two cases:
\begin{itemize}
\item If $\Gamma_0 =1$,
\begin{small}
\begin{align}
&\prob(X_0^1 \in E | Z_0, \tilde Z_0)\ind_{\{\Gamma_0 =1\}} = \prob(X_0^1 \in E | X_0^1, M_0^1) \ind_{\{\Gamma_0 =1\}} \notag \\
&= \prob(X_0^1 \in E | X_0^1) \ind_{\{\Gamma_0 =1\}} = \ind_{E}(X_0^1) \ind_{\{\Gamma_0 =1\}}.
\label{eq:gamma1_t0}
\end{align}
\end{small}
\item If $\Gamma_0 =0$,
\begin{small}
\begin{align}
&\prob(X_0^1 \in E | Z_0, \tilde Z_0) \ind_{\{\Gamma_0 =0\}} = \prob(X_0^1 \in E) \ind_{\{\Gamma_0 =0\}}  
\notag \\ &=\pi_{X_0^1}(E) \ind_{\{\Gamma_0 =0\}}. 
\end{align}
\end{small}
\end{itemize}
Hence, \eqref{nu_claim} holds at time $0$. Assume that \eqref{nu_claim} holds at time $t$.
This means that 
\begin{align}
\label{eq:induction_nu}
\prob^{\phi^{\diamond}_t}(dx_t^1 | H^0_{t}) = [\nu_{t}(H_{t}^0)] (dx_t^1).
\end{align}
At time $t+1$, since $\Gamma_{t+1} \in \{0,1\}$, consider two cases:
\begin{itemize}
\item If $\Gamma_{t+1} =1$, similar to \eqref{eq:gamma1_t0} we obtain
\begin{small}
\begin{align}
&\prob^{\phi^{\diamond}_{t+1}}(X_{t+1}^1 \in E | H^0_t, Z_{t+1},\tilde Z_{t+1}) 
\ind_{\{\Gamma_{t+1} =1\}} \notag \\
&=\ind_{E}(X_{t+1}^1) \ind_{\{\Gamma_{t+1} =1\}}
=[\nu_{t+1}(H_{t+1}^0)](E) \ind_{\{\Gamma_{t+1} =1\}}.
\end{align}
\end{small}
\item If $\Gamma_{t+1} =0$,
\begin{small}
\begin{align}
 &\hspace{-7mm} \prob^{\phi^{\diamond}_{t+1}}(X_{t+1}^1 \in E | H^0_t, Z_{t+1}, \tilde Z_{t+1}) \ind_{\{\Gamma_{t+1}^n =0\}} \notag \\
&\hspace{-7mm}=\prob^{\phi^{\diamond}_{t+1}}\big( f_t^1 (X_t^{0:1}, M_t^{0:1}, \phi_t^{prs}(H_t^0), W_t^1) \in E | H^0_t \big) \ind_{\{\Gamma_{t+1} =0\}} \notag \\
&\hspace{-7mm} 
=\sum_{m_t^1}
 \int\int\mathds{1}_{E}\big(f_t^1 (x_t^1, X_t^0, m_t^1, M_t^0, \phi_t^{prs}(H_t^0), w_t^1) \big)  \times
\notag \\
&\hspace{-4mm}\prob^{\phi^{\diamond}_{t}}(dx_t^1\vert H_t^0) \prob^{\phi^{\diamond}_{t}}(dw_t^1 \vert H_t^0) 
\prob^{\phi^{\diamond}_{t}}(M_t^1=m_t^1 \vert H_t^0)
\ind_{\{\Gamma_{t+1} =0\}} \notag \\
&\hspace{-7mm}
=\sum_{m_t^1}
 \int\int\mathds{1}_{E}\big(f_t^1 (X_t^0, x_t^1, M_t^0, m_t^1, \phi_t^{prs}(H_t^0), w_t^1) \big)  \times
\notag \\
&\hspace{0mm}
[\nu_{t}(H_{t}^0)] (dx_t^1) \pi_{W_t^1}(dw_t^1) \omega^{\tilde Z_{t}}(m_t^1)
\ind_{\{\Gamma_{t+1} =0\}}
\notag \\
&\hspace{-7mm}=[\nu_{t+1}(H_{t+1}^0)](E) \ind_{\{\Gamma_{t+1} =0\}},
\label{belief_proof_gamma_1}
\end{align}
\end{small}
where we have defined 
\begin{small}
\begin{align}
&f_t^1 (x_t^{0:1},m_t^{0:1}, u_t^{prs}, w_t^1)  \notag \\
&=
[D(m_t^{0:1})]_{1\bullet} \vecc \big(x_t^{0:1}, u_t^0, \rho_t(x_t^1,m_t^1) \big) + w_t^1.
\end{align}
\end{small}
Note that in \eqref{belief_proof_gamma_1}, the first equality is true due to the disintegration theorem \cite{Kallenberg} an the fact that $M_t^1$, $X_t^1$, and $W_t^1$ are conditionally independent given $H_{t}^0$ from Lemma \ref{lm:independence_of_state_mode}, and the second equality is true because of \eqref{eq:induction_nu} and Lemma \ref{lm:omega_equation}.
\end{itemize}
Hence, \eqref{nu_claim} holds at time $t+1$ and from \eqref{eq:theta_update},
we have $\theta_{t+1} = \psi_t(\theta_t, u_t^{prs}, x_t^0, m_t^0, \tilde z_t, z_{t+1})$. This completes the proof of Lemma \ref{lm:beliefupdate}.

\end{proof}

\section{Proof of Theorem \ref{thm:general_structure}}
\label{Proof_Thm_general_structure}
For any $\phi^{prs} \in \Phi^{prs}$ and any realization $h^0_t\in\mathcal{H}^0_t$, let the realization of the common belief $\Theta_t$ be
$\theta_t= \prob^{\phi^{prs}_{0:t-1}} (dx_t^1 \vert h_t^0)$ defined by Lemma \ref{lm:beliefupdate}.
Suppose the prescription strategy $\phi^{prs*} \in \Phi^{prs}$ achieves the minimum of \eqref{eq:DP_V} for $x_{t}^0, m_{t}^0, \theta_{t}, \tilde z_{t}$, $t=0,\ldots,T$, and let $u^{prs*}_t=(u^{0*}_t, \rho_t^*) = \phi^{prs*}(h_t^0)$ for any realization $h^0_t\in\mathcal{H}^0_t$.

We prove by induction that $V_{t}(x_{t}^0, m_{t}^0, \theta_{t}, \tilde z_{t})$ is a measurable function with respect to $h_t^0$, and for any $h_t^0 \in \mathcal{H}_t^0$ and for any $ \phi^{\diamond}_t := \phi^{prs}_{0:t-1} \in \Phi^{prs}$
we have

\vspace{-4mm}
\begin{small}
\begin{align}
 &\ee^{\phi'_t}  \left[\sum_{s=t}^T c_s^{prs}(X_s^{0:1}, M_s^{0:1}, U^{prs}_s) \middle| h_t^0 \right]
 \nonumber\\
&= V_{t}\big(
x_{t}^0, m_{t}^0, \prob^{\phi^{\diamond}_t} (dx_t^1 \vert h_t^0) , \tilde z_{t}\big)
\label{eq:Vinduction_part1}
\\
& \leq  \ee^{\phi^{prs}}\left[\sum_{s=t}^T c_s^{prs}(X_s^{0:1}, M_s^{0:1}, U^{prs}_s) \middle| h_t^0\right]
\label{eq:Vinduction_part2}
\end{align}
\end{small}
where $\phi'_t= \{\phi^{\diamond}_t,  \phi^{prs*}_{t:T}\}$.
Note that the above equation at $t=0$ gives the optimality of $\phi^{prs*}$ for Problem \ref{problem:equivalent}.

We first consider \eqref{eq:Vinduction_part1}.
At $T+1$, \eqref{eq:Vinduction_part1} is true (all terms are defined to be $0$ at $T+1$). 
Assume $V_{t+1}(x_{t+1}^0, m_{t+1}^0, \theta_{t+1}, \tilde z_{t+1})$
is a measurable function with respect to $h_{t+1}^0$ and
\eqref{eq:Vinduction_part1} is true at $t+1$, that is, for any $h_{t+1}^0 \in \mathcal{H}_{t+1}^0$ and for any $\phi^{\diamond}_{t+1} \in \Phi^{prs}$

\begin{small}
\begin{align}
 &\ee^{\phi'_{t+1}}  \left[\sum_{s=t+1}^T c_s^{prs}(X_s^{0:1}, M_s^{0:1}, U^{prs}_s) \middle| h_{t+1}^0 \right]
 \nonumber\\
&= V_{t+1}\Big(
x_{t+1}^0, m_{t+1}^0, \prob^{\phi^{\diamond}_{t+1}} (dx_{t+1}^1 \vert h_{t+1}^0) , \tilde z_{t+1}\Big)
\label{eq:induction_hypothesis_p1}
\end{align}
\end{small}
where $\phi'_{t+1}= \{\phi^{\diamond}_{t+1},  \phi^{prs*}_{t+1:T}\}$. Since \eqref{eq:induction_hypothesis_p1} is true for any $\phi^{\diamond}_{t+1} \in \Phi^{prs}$, by choosing $\phi^{prs}_t$ to be $\phi^{prs*}_t$, we get $\phi'_{t+1}= \{\phi^{\diamond}_{t+1},  \phi^{prs*}_{t+1:T}\} = 
\{\phi^{\diamond}_{t},  \phi^{prs*}_{t:T}\} =  \phi'_{t}$. Then, from \eqref{eq:induction_hypothesis_p1}, we get,
\begin{small}
\begin{align}
 &\ee^{\phi'_{t+1}}  \left[\sum_{s=t+1}^T c_s^{prs}(X_s^{0:1}, M_s^{0:1}, U^{prs}_s) \middle| h_{t+1}^0 \right]
 \nonumber\\
&= V_{t+1}\Big(x_{t+1}^0, m_{t+1}^0, \prob^{\phi^{\diamond}_{t}, \phi^{prs*}_{t}} (dx_{t+1}^1 \vert h_{t+1}^0) , \tilde z_{t+1}\Big).
\label{eq:induction_hypothesis_new_p1}
\end{align}
\end{small}

At time $t$, from the tower property of conditional expectation we have
\begin{small}
\begin{align}
&\ee^{\phi'_t}  \left[\sum_{s=t}^T c_s^{prs}(X_s^{0:1}, M_s^{0:1}, U^{prs}_s) \middle| h_t^0 \right]
\nonumber\\
&=\ee^{\phi'_t}\left[ c_t^{prs}(X_t^{0:1}, M_t^{0:1}, U^{prs}_t) \middle| h_t^0 \right] \nonumber\\
&+\ee^{\phi'_t}\left[\ee^{\phi'_t}\left[\sum_{s=t+1}^T c_s^{prs}(X_s^{0:1}, M_s^{0:1}, U^{prs}_s) \middle| H_{t+1}^0 \right]\middle| h_t^0 \right].
\label{eq:Vinduction_part1_zero}
\end{align}
\end{small}
Note that the first term in \eqref{eq:Vinduction_part1_zero} is equal to

\vspace{-4mm}
\begin{small}
\begin{align}
&\ee^{\phi'_t}\left[ c_t^{prs}(X_t^{0:1}, M_t^{0:1}, U^{prs}_t) \middle| h_t^0 \right] \notag \\
&=\sum_{m_t^1 \in \mathcal{M}^1}\int  c_t^{prs}(x_t^{0:1}, m_t^{0:1}, u_t^{prs*}) \theta_t(dx_t^1) \omega^{\tilde z_{t}}(m_t^1) \label{eq:Vinduction_part1_first_1}  \\
& =:\ee \Big[ c_t^{prs}(x_t^0, X_t^{1}, m_t^0, M_t^{1}, u^{prs}_t)  \Big| x_t^0, m_t^0, \theta_t, \tilde z_t, u_t^{prs*} \Big]
\label{eq:Vinduction_part1_first} 
\end{align}
\end{small}
where the first equality is true because of Lemma \ref{lm:independence_of_state_mode}. Furthermore, we can write \eqref{eq:Vinduction_part1_first_1} as \eqref{eq:Vinduction_part1_first} because of Lemma \ref{lm:independence_of_state_mode} where in \eqref{eq:Vinduction_part1_first}, $X_t^1$ is a random vector distributed according to $\theta_t$.

According to \eqref{eq:induction_hypothesis_new_p1}, the second term in \eqref{eq:Vinduction_part1_zero} can be written as

\vspace{-4mm}
\begin{small}
\begin{align}
&\hspace{0mm} \ee^{\phi'_t}\left[
V_{t+1}\big(X_{t+1}^0, M_{t+1}^0, \prob^{\phi^{\diamond}_{t}, \phi^{prs*}_{t}} (dx_{t+1}^1 \vert H_{t+1}^0) , \tilde Z_{t+1}\big)
\middle| h_t^0 \right] \notag \\
&\hspace{0mm} =\ee^{\phi'_t}\left[
V_{t+1}\big(X_{t+1}^0, M_{t+1}^0, \psi_t^{\circ*}(Z_{t+1}) , \tilde Z_{t+1}\big)
\middle| h_t^0 \right]
\notag  \\
 &\hspace{0mm}= 
 \sum_{\gamma_{t+1} \in \{0,1\}} 
 \ee^{\phi'_t}\Big[
V_{t+1}\big(X_{t+1}^0, M_{t+1}^0, \psi_t^{\circ*}(Z_{t+1}) , \tilde Z_{t+1}\big)
\notag \\
& \hspace{4cm} \Big\vert h_t^0, \Gamma_{t+1}=\gamma_{t+1} \Big] p(\gamma_{t+1}),
\label{eq:Vinduction_part1_second_p1} 
\end{align}
\end{small}
where we have defined 
\begin{small}
\begin{align}
\psi_t^{\circ *}(Z_{t+1}) := \psi_t(\theta_t, u_t^{prs*}, x_t^0, m_t^0, \tilde z_t, Z_{t+1}).
\label{psi_circ}
\end{align}
\end{small}
Note that 
\begin{small}
\begin{align}
&
 \ee^{\phi'_t}\Big[
V_{t+1}\big(X_{t+1}^0, M_{t+1}^0, \psi_t^{\circ *}(\emptyset), \emptyset\big)
\Big\vert h_t^0, \Gamma_{t+1}= 0 \Big] = \sum_{m_{t+1}^0 \in \mathcal{M}^0}
\notag \\
& \hspace{5mm} \int V_{t+1} \Big(x_{t+1}^0, m_{t+1}^0, \alpha_t^{1*}, \emptyset \Big)  \alpha_t^{0*}(dx_{t+1}^0) \pi_{M^0}(m_{t+1}^0).
\label{eq:Vinduction_part1_second_p2} 
\end{align}
\end{small}
In \eqref{eq:Vinduction_part1_second_p2}, $\alpha_t^{1*} = \psi_t^{\circ}(\emptyset)$ and $\alpha_t^{0*} = \tilde{\psi}_t(x_t^0, m_t^0, u_t^{0*})$ where for any $E \subset \R^{d_X^0}$ we have 
\begin{align}
&[\tilde{\psi}_t(x_t^0, m_t^0, u_t^0)] (E):=  \prob^{\phi^{\diamond}_t}(X_{t+1}^0 \in E | h^0_{t}) \notag \\
&= \prob \big(A^{00}(m_t^0) x_t^0 + B^{00}(m_t^0) u_t^0  + W_t^0 \in E \big)
\label{alpha_0}
\end{align}
and the last equality of \eqref{eq:Vinduction_part1_second_p2} is true because $\prob^{\phi'_t} (dx_{t+1}^1 \vert h_t^0, \Gamma_{t+1}=0) = \alpha_t^{1*}(dx_{t+1}^1)$ from Lemma \ref{lm:beliefupdate}. Furthermore,

\vspace{-4mm}
\begin{small}
\begin{align}
&\hspace{0mm}
 \ee^{\phi'_t}\Big[
V_{t+1}\big(X_{t+1}^0, M_{t+1}^0, \psi_t^{\circ *}(X_{t+1}^1), M_{t+1}^1\big)
\Big\vert h_t^0, \Gamma_{t+1}= 1 \Big]  \notag \\
&\hspace{0mm} =
 \ee^{\phi'_t}\Big[
V_{t+1}\big(X_{t+1}^0, M_{t+1}^0, \psi_t^{\circ *}(X_{t+1}^1), M_{t+1}^1\big)
\Big\vert h_t^0, \Gamma_{t+1}= 0 \Big]  \notag \\
&= \sum_{
\substack{m_{t+1}^0 \in \mathcal{M}^0\\ m_{t+1}^1 \in \mathcal{M}^1}}
\int \int  V_{t+1} \Big(x_{t+1}^0, m_{t+1}^0, \delta(x_{t+1}^1), m_{t+1}^1 \Big)
 \notag \\
& \hspace{20mm} 
 \times \prod_{n\in \{0,1\}}\alpha_t^{n*}(dx_{t+1}^{n}) \pi_{M^n}(m_{t+1}^n),
\label{eq:Vinduction_part1_second_p3} 
\end{align}
\end{small}
The first equality in \eqref{eq:Vinduction_part1_second_p3} is true because $X_{t+1}^{1:2}$ and $M_{t+1}^{1:2}$ are independent of $\Gamma_{t+1}$.

Now, by combining \eqref{eq:Vinduction_part1_second_p2} and \eqref{eq:Vinduction_part1_second_p3},
 \eqref{eq:Vinduction_part1_second_p1}can be written as

\begin{small}
\begin{align}
&\ee^{\phi'_t}\left[
V_{t+1}\big(X_{t+1}^0, M_{t+1}^0, \psi_t^{\circ*}(Z_{t+1}) , \tilde Z_{t+1}\big)
\middle| h_t^0 \right] \notag \\
& = \ee \Big[V_{t+1}\big(X_{t+1}^0, M_{t+1}^0, \psi_t^{\circ*}(Z_{t+1}), \tilde Z_{t+1}\big)  \Big| x_t^0, m_t^0, \theta_t, \tilde z_t, u_t^{prs*} \Big].
\label{eq:Vinduction_part1_second_p4} 
\end{align}
\end{small}


Now, from \eqref{eq:Vinduction_part1_first} and \eqref{eq:Vinduction_part1_second_p4}, the right hand side of \eqref{eq:Vinduction_part1_zero} is $V_t(x_{t}^0, m_{t}^0, \theta_{t}, \tilde z_{t}) $ from the definition of the value function \eqref{eq:DP_V}.
Hence, \eqref{eq:Vinduction_part1} is true at time $t$. The measurability  of $V_t(x_{t}^0, m_{t}^0, \theta_t, \tilde z_{t}) $ with respect to $h_t^0$ is also resulted from the fact that 
$V_{t}(x_{t}^0, m_{t}^0, \prob^{\phi^{\diamond}_t} (dx_t^1 \vert h_t^0) , \tilde z_{t})$ is equal to the conditional expectation $\ee^{\phi'_t}\left[\sum_{s=t}^T c_s^{prs}(X_s^{0:1}, M_s^{0:1}, U^{prs}_s) \middle| h^0_t\right]$ which is measurable with respect to $h^0_t$.

Now let's consider \eqref{eq:Vinduction_part2}.
At $T+1$, \eqref{eq:Vinduction_part2} is true (all terms are defined to be $0$ at $T+1$). 
Assume \eqref{eq:Vinduction_part2} is true at $t+1$.
Let $u^{prs}_t=(u^{0}_t,\rho_t) = \phi^{prs}(h^0_t)$. 
Following an argument similar to that of \eqref{eq:Vinduction_part1_zero}-\eqref{eq:Vinduction_part1_second_p3},

\begin{small}
\begin{align}
&\ee^{\phi^{prs}}\left[\sum_{s=t}^T c_s(X^{0:N}_s,U^{0:N}_s)  \middle| h_t^0 \right] \geq
\notag\\
&\ee \Big[ c_t^{prs}(x_t^0, X_t^{1}, m_t^0, M_t^{1}, u^{prs}_t)  \Big| x_t^0, m_t^0, \theta_t, \tilde z_t, u_t^{prs} \Big] \notag \\
&+  \ee \Big[V_{t+1}\big(X_{t+1}^0, M_{t+1}^0, \psi_t^{\circ}(Z_{t+1}), \tilde Z_{t+1}\big)  \Big| x_t^0, m_t^0, \theta_t, \tilde z_t, u_t^{prs} \Big]
\notag\\
&\geq V_t(x_{t}^0, m_{t}^0, \theta_{t}, \tilde z_{t}).
\label{eq:Vinduction_part2_proof}
\end{align}
\end{small}
where the last inequality follows from the definition of the value function \eqref{eq:DP_V}.
This completes the proof of the induction step, and the proof of the theorem.

\section{Proof of Lemma \ref{decomposition_equivalence}}
\label{Proof_lemma_decomposition}
To show \eqref{set_equivalence}, let $\bar{\mathcal{P}}$ denote the set $\bar{\mathcal{P}} = \{\bar q_t \circ h_2 + \tilde q_t:  \bar q_t \in \bar{\mathcal{Q}}, \tilde q_t \in \tilde{\mathcal{Q}}(\theta_t)  \}$. Then, we want to show that for any $\theta_t \in \Delta(\R^{d_X^1})$, $\mathcal{\bar P} =\mathcal{P}$ where $ \mathcal{P}=\{\rho: \R^{d_X^1} \times \mathcal{M}^1 \to \R^{d_U^1} \text{ measurable} \}$. 

First, assume $\rho_t \in \mathcal{P}$. For any $\theta_t \in \Delta(\R^{d_X^1})$, define 
$\bar q_t(\cdot) = \int \rho_t (x_t^1,\cdot) \theta_t(dx_t^1)$ and 
$\tilde q_t = \rho_t - \bar q_t \circ h_2$. Then, we have $\rho_t = \bar q_t \circ h_2 + \tilde q_t$.
Note that $\bar q_t \in \mathcal{\bar Q}$ and since $\rho_t$ is measurable, $\tilde q_t$ is measurable. Furthermore, $\int \tilde q_t(x_t^1,m_t^1) \theta_t(dx_t^1) =  \int \rho_t (x_t^1,m_t^1) \theta_t(dx_t^1) -  \int \rho_t (x_t^1,m_t^1) \theta_t(dx_t^1) = 0$ for any $m_t^1 \in \mathcal{M}^1$. Hence,
$\tilde q_t \in \mathcal{\tilde Q}(\theta_t)$. This concludes that $\rho_t \in \mathcal{\bar P}$.

For the reverse direction, assume that $\rho_t \in \mathcal{\bar P}$. Then, it can be written as $\rho_t = \bar q_t \circ h_2 + \tilde q_t$ where  $\bar q_t \in \bar{\mathcal{Q}}$ and $\tilde q_t \in \tilde{\mathcal{Q}}(\theta_t)$. Since, $\bar q_t$ and $\tilde q_t$ are measurable, $\rho_t$ is a measurable function from $\R^{d_X^1} \times \mathcal{M}^1$ to $\R^{d_U^1}$ and hence, $\rho_t \in \mathcal{P}$. This completes the proof.

\section{Proof of Theorem \ref{thm:Sol_mode}}
\label{Proof_theorem_Sol_mode}
The proof is done by induction. 
\begin{itemize}
\item At time $T+1$:
\end{itemize}
It is clear that \eqref{eq:Vt_mode} is true because $P_{T+1}(m_{T+1}^0,\tilde z_{T+1}) =  \tilde{P}_{T+1}(m_{T+1}^0,\tilde z_{T+1}) = \mathbf{0}$ for any $m_{T+1}^0 \in \mathcal{M}^0$, and $\tilde z_{T+1} \in \tilde{\mathcal{M}}^1$, and $e_{T+1}=0$.
\begin{itemize}
\item At time $t+1$:
\end{itemize}
Suppose \eqref{eq:Vt_mode} is true, that is, for any $x_t^0 \in \R^{d_X^0}$, $m_t^0 \in \mathcal{M}^0$, $\theta_t \in \Delta(\R^{d_X^1})$, and $\tilde z_t \in \tilde{\mathcal{M}}^1$,
\begin{small}
\begin{align}
V_{t+1}(&x_{t+1}^0, m_{t+1}^0, \theta_{t+1}, \tilde z_{t+1}) =  \notag \\
& QF \Big(P_{t+1}(m_{t+1}^0,\tilde z_{t+1}), \vecc \big(x^0_{t+1},\mu(\theta_{t+1}) \big) \Big) \notag \\
& + \tr \Big( \tilde{P}_{t+1}(m_{t+1}^0,\tilde z_{t+1}) \cov(\theta_{t+1}) \Big) + e_{t+1},
\label{eq:V_t_1}
\end{align}
\end{small}
and the matrices $P_{t+1}(m_{t+1}^0,\tilde z_{t+1})$ and $\tilde{P}_{t+1}(m_{t+1}^0,\tilde z_{t+1})$ are all positive semi-definite (PSD) for any $m_{t+1}^0 \in \mathcal{M}^0$, and $\tilde z_{t+1} \in \tilde{\mathcal{M}}^1$.

\begin{itemize}
\item At time $t$:
\end{itemize}

Let's now compute the value function at $t$ given by \eqref{eq:DP_V} in Theorem \ref{thm:general_structure}. In order to do so, we need to calculate 

\begin{small}
\begin{align}
&\mathbb{T} =  
\underbrace{
\ee \Big[ c_t^{prs}(x_t^0, X_t^{1}, m_t^0, M_t^{1}, u^{prs}_t)  \Big| x_t^0, m_t^0, \theta_t, \tilde z_t, u_t^{prs} \Big] }
_{\mathbb{T}_1(\tilde z_t)}
\notag \\
&+  \underbrace{
\ee \Big[V_{t+1}\big(X_{t+1}^0, M_{t+1}^0, \psi_t^{\circ}(Z_{t+1}), \tilde Z_{t+1}\big)  \Big| x_t^0, m_t^0, \theta_t, \tilde z_t, u_t^{prs} \Big]
}
_{\mathbb{T}_2(\tilde z_t)}
\end{align}
\end{small}
where the first and second term are as defined in \eqref{eq:Vinduction_part1_first} and \eqref{eq:Vinduction_part1_second_p4}, respectively and $\psi_t^{\circ}(Z_{t+1}) = \psi_t(\theta_t, u_t^{prs}, x_t^0, m_t^0, \tilde z_t, Z_{t+1})$.

To this end, we consider two following cases.
\subsection{$\tilde z_t = \emptyset$}
This corresponds to the case that $\gamma_t =0$ from \eqref{Model:channel}. In the following we calculate $\mathbb{T}_1(\emptyset)$ and $\mathbb{T}_2(\emptyset)$.

\underline{\textbf{Calculating $\mathbb{T}_1(\emptyset)$:}}

Let $S_t^{\theta_t,m_t^1} := \vecc(x_t^0, X^{\theta_t},u^0_t, \bar q_t(m_t^1)+ \tilde q_t(X^{\theta_t},m_t^1) )$ where $X^{\theta_t}$ is a random vector with distribution $\theta_t$.
Then, according to \eqref{cost_structure} and \eqref{cost_of_new_problem},

\vspace{-4mm}
\begin{small}
\begin{align}
\mathbb{T}_1(\emptyset) &= 
\ee \Big[ c_t^{prs}(x_t^0, X_t^{1}, m_t^0, M_t^{1}, u^{prs}_t)  \Big| x_t^0, m_t^0, \theta_t, \emptyset, u_t^{prs} \Big]  \notag \\
&=\sum_{m_t^1 \in \mathcal{M}^1}\int  c_t^{prs}(x_t^{0:1}, m_t^{0:1}, u_t^{prs}) \theta_t(dx_t^1) \pi_{M^1}(m_t^1)
\notag \\
&=
\sum_{m_t^1 \in \mathcal{M}^1} \ee \Big[c_t^{prs}(x_t^0, X^{\theta_t}, m_t^{0:1}, u_t^{prs}) \Big] \pi_{M^1}(m_t^1)
\notag \\
&=  \sum_{m_t^1 \in \mathcal{M}^1} \ee \Big[QF\left(C_t(m_t^{0:1}), S_t^{\theta_t,m_t^1} \right) \Big]
\pi_{M^1}(m_t^1) \notag \\
&=  \sum_{m_t^1 \in \mathcal{M}^1} \Big[ 
QF\left(C_t(m_t^{0:1}), \ee[S_t^{\theta_t,m_t^1}]\right) \notag \\
&+  \tr\left(C_t(m_t^{0:1})\cov(S_t^{\theta_t,m_t^1}) \right)
 \Big] \pi_{M^1}(m_t^1),
\label{eq:Vterm_inst_cost}
\end{align}
\end{small}
where the second equality is true because of \eqref{eq:Vinduction_part1_first}.

Note that in the minimization of Theorem \ref{thm:structure}, we are looking for only $\tilde q_t \in \mathcal{\tilde{Q}}(\theta_t)$. Hence, according to \eqref{Q_set}, $\ee[\tilde q_t (X^{\theta_t},m_t^1)] =0$ and

\vspace{-4mm}
\begin{small}
\begin{align}
\ee[S_t^{\theta_t,m_t^1}] &= \vecc(x_t^0, \mu(\theta_t),u^0_t, \bar q_t(m_t^1))
\label{S_t_m_mean} \\
\cov(S_t^{\theta_t,m_t^1}) &= \cov \big(\vecc(0, X^{\theta_t},0, \tilde q_t(X^{\theta_t},m_t^1) ) \big).
\label{S_t_m_cov} 
\end{align}
\end{small}
Let $\bar S_t^{\theta_t} = \vecc(x_t^0, \mu(\theta_t),u^0_t, \bar q_t(1), \ldots, \bar q_t(\kappa^1))$. Then, 

\vspace{-4mm}
\begin{small}
\begin{align}
\ee[S_t^{\theta_t,m_t^1}] = L_{m_t^1}\bar S_t^{\theta_t},
\label{S_to_bar_S}
\end{align}
\end{small}
where $L_{m_t^1}$ is as defined in \eqref{Lm_matrix}.

Furthermore, according to \eqref{S_t_m_cov}, we can write 

\vspace{-4mm}
\begin{small}
\begin{align}
&\tr\left(C_t(m_t^{0:1})\cov(S_t^{\theta_t,m_t^1}) \right)
= \tr\left(C_t^{11} (m_t^{0:1}) \cov(\tilde S_t^{\theta_t,m_t^1}) \right),
\label{2nd_T_1}
\end{align}
\end{small}
where we defined $\tilde S_t^{\theta_t,m_t^1} := \vecc(X^{\theta_t}, \tilde q_t(X^{\theta_t},m_t^1))$.
Now, using \eqref{S_to_bar_S} and \eqref{2nd_T_1}, \eqref{eq:Vterm_inst_cost} can be written as,
\begin{small}
\begin{align}
\mathbb{T}_1(\emptyset)&=  
QF\left(C_t^{\emptyset}(m_t^{0}), \bar S_t^{\theta_t}\right) \notag \\
&+ 
\sum_{m_t^1 \in \mathcal{M}^1} \tr\left(C_t^{11} (m_t^{0:1}) \cov(\tilde S_t^{\theta_t,m_t^1}) \right)
 \pi_{M^1}(m_t^1).
\label{T_1_empty}
\end{align}
\end{small}

\underline{\textbf{Calculating $\mathbb{T}_2(\emptyset)$:}}

We first calculate $\mathbb{T}_2(\tilde z_t) $ for any $\tilde z_t \in \mathcal{\bar M}^1$, and then simplify it for the case $\tilde z_t = \emptyset$. According to \eqref{eq:Vinduction_part1_second_p4},
\begin{small}
\begin{align}
&\mathbb{T}_2(\tilde z_t) = \sum_{\gamma_{t+1}\in\{0,1\}}p(\gamma_{t+1})
 \sum_{
\substack{m_{t+1}^0 \in \mathcal{M}^0\\ m_{t+1}^1 \in \mathcal{M}^1}}
 \notag \\
& \hspace{1.5cm} \int \int  \underbrace{V_{t+1} \Big(x_{t+1}^0, m_{t+1}^0, \textit{NB}(\gamma_{t+1},\alpha_t^1,x_{t+1}^1), \tilde z_{t+1} \Big)}_{\mathbb{T}_3}
 \notag \\
& \hspace{1.5cm}\times \prod_{n\in \{0,1\}}\alpha_t^n(dx_{t+1}^n) \pi_{M^n}(m_{t+1}^n).
\label{T_2_empty}
\end{align}
\end{small}
where $\alpha_t^1 = \psi_t^{\circ}(\emptyset) =\psi_t(\theta_t, u_t^{prs}, x_t^0, m_t^0, \tilde z_t, \emptyset)$, $\alpha_t^0 =\tilde{\psi}_t(x_t^0, m_t^0, u_t^0)$ as defined in \eqref{alpha_0}, and we defined $\textit{NB}(\gamma_{t+1}, \alpha_t^1, x_{t+1}^1):= (1-\gamma_{t+1})\alpha_t^1 + \gamma_{t+1} \delta(x_{t+1}^1)$. Further, we can write $x_{t+1}^0 = \mu(\alpha_t^0) + \big (x_{t+1}^0 - \mu(\alpha_t^0)  \big)$.
Considering these and after some algebra, we can get, 

\vspace{-4mm}
\begin{small}
\begin{align}
&
\int \int
 \mathbb{T}_3 
 \prod_{n\in \{0,1\}}\alpha_t^n(dx_{t+1}^n) 
  \notag \\
&=
\gamma_{t+1}QF\Big(P_{t+1}(m_{t+1}^0, m_{t+1}^1), \vecc \big(\mu(\alpha_t^0), \mu(\alpha_t^1)  \big)  \Big) \notag \\
&\hspace{0cm}+ (1-\gamma_{t+1})QF\Big(P_{t+1}(m_{t+1}^0, \emptyset), \vecc \big(\mu(\alpha_t^0), \mu(\alpha_t^1)  \big)  \Big) \notag \\
& \hspace{0cm}+ \gamma_{t+1} \tr \big(P_{t+1}^{00}(m_{t+1}^0,m_{t+1}^1) \cov(\alpha_t^0) \big) \notag \\
&\hspace{0cm}+ (1-\gamma_{t+1}) \tr \big(P_{t+1}^{00}(m_{t+1}^0,\emptyset) \cov(\alpha_t^0) \big)
 \notag \\
& \hspace{0cm} + \gamma_{t+1} \tr \big(P_{t+1}^{11}(m_{t+1}^0,m_{t+1}) \cov(\alpha_t^1) \big) \notag \\
&\hspace{0cm} +  (1-\gamma_{t+1}) \tr \big(\tilde P_{t+1}(m_{t+1}^0,\emptyset) \cov(\alpha_t^1) \big) + e_{t+1}.
\label{eq:Vterm_4_mode}
\end{align}
\end{small}
Note that $\alpha_t^1$ depends on $\tilde z_t$ and $\theta_t$. To make this dependence apparent, we write $\alpha_t^1$ as $\psi_t^{\triangleleft}(\tilde z_t, \theta_t)$. Then, from \eqref{eq:Vterm_4_mode} and using operators $\Pi$ and $\Psi$ defined in \eqref{operator_Pi} and \eqref{operator_Psi}, \eqref{T_2_empty} can be written as

\vspace{-4mm}
\begin{small}
\begin{align}
&\mathbb{T}_2 (\tilde z_t)= 
 \underbrace{QF\Big( \Pi(P_{t+1}), \vecc \big(\mu(\alpha_t^0), \mu(\psi_t^{\triangleleft}(\tilde z_t, \theta_t))  \big)  \Big)
 }_{\mathbb{T}_4(\tilde z_t)} \notag \\
& +\underbrace{ \tr \big(\Pi(P_{t+1}^{00}) \cov(\alpha_t^0)\big )}_{\mathbb{T}_5}
  \notag \\
& + \underbrace{ \tr \Big( \Psi(\tilde P_{t+1}, P_{t+1}^{11})\cov(\psi_t^{\triangleleft}(\tilde z_t, \theta_t))\Big)}_{\mathbb{T}_6(\tilde z_t)} +e_{t+1}.
\label{eq:Vterm_5_mode}
\end{align}
\end{small}

Now, to calculate $\mathbb{T}_2 (\emptyset)$ from \eqref{eq:Vterm_5_mode}, we need to calculate mean and covariance of $\alpha_t^0$ and $\psi_t^{\triangleleft}(\emptyset, \theta_t)$.

\underline{Calculating mean and covariance of $\alpha_t^0$ and $\psi_t^{\triangleleft}(\emptyset, \theta_t)$:}

Let $S_t^{\theta_t,M_t^1} := \vecc(x_t^0, X^{\theta_t},u^0_t, \bar q_t(M_t^1)+ \tilde q_t(X^{\theta_t},M_t^1) )$ where $X^{\theta_t}$ is a random vector independent of $M^1_t$ and with distribution $\theta_t$. Then, we define $Y_t^{\theta_t,M_t^1}: = [D(m_t^0,M_t^1)]_{1\bullet} S_t^{\theta_t,M_t^1} + W_t^1$ and $\tilde Y_t: = A^{00}(m_t^0) x_t^0 + B^{00}(m_t^0) u_t^0  + W_t^0$. From \eqref{eq:psit} in Lemma \ref{lm:beliefupdate}, we know that $Y_t^{\theta_t,M_t^1}$ has distribution $\psi_t^{\triangleleft}(\emptyset, \theta_t)$, and furthermore, from \eqref{alpha_0} we know that $\tilde Y_t$ has distribution $\alpha_t^0$. Then, using the fact that from  \eqref{S_to_bar_S}, we can write $\ee[S_t^{\theta_t,m_t^1}] = L_{m_t^1}\bar S_t^{\theta_t}$, we have

\vspace{-4mm}
\begin{small}
\begin{align}
\label{mean_alpha_0}
\mu(\alpha_t^0) &= \ee[\tilde Y_t]  = A^{00}(m_t^0) x_t^0 + B^{00}(m_t^0) u_t^0  \notag \\
&= \hspace{-2mm} \sum_{m_t^1 \in \mathcal{M}^1} 
[D(m_t^{0:1})]_{0\bullet} L_{m_t^1}\bar S_t^{\theta_t} \pi_{M^1}(m_t^1),  \\
\label{cov_alpha_0}
\cov(\alpha_t^0) &= \cov(\tilde Y_t) = \cov(W_t^0),  \\
\mu(\psi_t^{\triangleleft}(\emptyset, \theta_t)) &= \ee[Y_t^{\theta_t,M_t^1}]  \notag \\
&= 
\hspace{-2mm}
\sum_{m_t^1 \in \mathcal{M}^1} 
[D(m_t^{0:1})]_{1\bullet} L_{m_t^1}\bar S_t^{\theta_t} \pi_{M^1}(m_t^1),   \\
\cov(\psi_t^{\triangleleft}(\emptyset, \theta_t)) &= \cov(Y_t^{\theta_t,M_t^1}) = \ee[\cov(Y_t^{\theta_t,M_t^1}\vert M_t^1)] \notag \\
& \hspace{22mm} + \cov(\ee[Y_t^{\theta_t,M_t^1}\vert M_t^1]),
\label{cov_alpha_1}
\end{align}
\end{small}
where the last equality of \eqref{mean_alpha_0} is true because
$[D(m_t^{0:1})]_{0\bullet} = \begin{bmatrix}
A^{00}(m_t^0) & \mathbf{0} & B^{00}(m_t^0) & \mathbf{0}
\end{bmatrix}$ and further, \eqref{cov_alpha_1} is true because of "Law of total variance" and
\begin{small}
\begin{align}
\label{cov_alpha_1_1}
&\ee[\cov(Y_t^{\theta_t,M_t^1}\vert M_t^1)]  = 
\sum_{m_t^1 \in \mathcal{M}^1} 
\hspace{-2mm}
\cov(Y_t^{\theta_t,m_t^1}\vert m_t^1)
\pi_{M^1}(m_t^1)=
\notag \\
& \hspace{-1mm} \sum_{m_t^1 \in \mathcal{M}^1}
\hspace{-2mm}
D^{11}(m_t^{0:1})\cov(\tilde S_t^{\theta_t,m_t^1}) D^{11}(m_t^{0:1})^{\tp}
\pi_{M^1}(m_t^1)+ \cov(W_t^1), \\
&\cov(\ee[Y_t^{\theta_t,M_t^1}\vert M_t^1]) = 
\cov \big([D(m_t^0,M_t^1)]_{1\bullet} \ee[S_t^{\theta_t,M_t^1}\vert M_t^1]\big) =
\notag \\
& \sum_{\substack{m_t^1 \in \mathcal{M}^1}}
[D(m_t^{0:1}) ]_{1\bullet} L_{m_t^1}\bar S_t^{\theta_t} (L_{m_t^1}\bar S_t^{\theta_t})^{\tp}  [D(m_t^{0:1}) ]_{1\bullet}^{\tp} \pi_{M^1}(m_t^1)
\notag \\
&- [\tilde D^{\emptyset} (m_t^0)]_{1\bullet} [\tilde D^{\emptyset} (m_t^0)]_{1\bullet}^{\tp},
\label{cov_alpha_1_2}
\end{align}
\end{small}
where we defined $\tilde D^{\emptyset} (m_t^0) = \sum_{\substack{m_t^1 \in \mathcal{M}^1}}
[D(m_t^{0:1}) ]_{1\bullet} L_{m_t^1}\bar S_t^{\theta_t} \pi_{M^1}(m_t^1)$.

\vspace{4mm}
Now that we have calculated the mean and covariance of $\alpha_t^0$ and $\psi_t^{\triangleleft}(\emptyset, \theta_t)$ in \eqref{mean_alpha_0}-\eqref{cov_alpha_1_2}, using the matrices $E^{\emptyset}_t, F^{\emptyset}_t$ defined in \eqref{eq:recursion_E_t} and \eqref{eq:recursion_F_t}, $\mathbb{T}_4(\emptyset) $, $\mathbb{T}_5$, and $\mathbb{T}_6(\emptyset) $ can be respectively written as follows:

\vspace{-4mm}
\begin{small}
\begin{align}
\mathbb{T}_4(\emptyset) &= QF \big( E^{\emptyset}_t(m_t^0), \bar S_t^{\theta_t} \big),
 \label{T_3} \\
\mathbb{T}_5 &=   \tr \Big(\Pi(P_{t+1}^{00}) \cov(W_t^0)\Big ),
  \label{T_4} \\
 \mathbb{T}_6(\emptyset)&= \sum_{m_t^1 \in \mathcal{M}^1} \tr \Big( G_t(m_t^{0:1})  \cov(\tilde S_t^{\theta_t,m_t^1}) \Big) 
\pi_{M^1}(m_t^1)  
 \notag \\
& + QF \Big( F_t^{\emptyset}(m_t^0), \bar S_t^{\theta_t} \Big) + \tr \Big( \Psi(\tilde P_{t+1}, P_{t+1}^{11}) \cov(W_t^1)\Big).
 \label{T_5}
\end{align}
\end{small}
where we have defined $G_t(m_t^{0:1}) = D^{11}(m^{0:1})^{\tp} \Psi(\tilde P_{t+1}, P_{t+1}^{11}) D^{11}(m^{0:1})$. Now, considering \eqref{T_3}-\eqref{T_5}, \eqref{eq:Vterm_5_mode} can be written as

\begin{small}
\begin{align}
&\mathbb{T}_2 (\emptyset) = QF \Big(E_t^{\emptyset}(m_t^0) + F_t^{\emptyset}(m_t^0), \bar S_t^{\theta_t} \Big) + \notag \\ 
&\sum_{m_t^1 \in \mathcal{M}^1} \tr \Big( G_t(m_t^{0:1})  \cov(\tilde S_t^{\theta_t,m_t^1}) \Big) \pi_{M^1}(m_t^1)    + e_t,
\label{T_6_empty}
\end{align}
\end{small}
where
\begin{small}
\begin{align}
e_t &= e_{t+1} + \tr \Big(\Pi(P_{t+1}^{00}) \cov(W_t^0)\Big ) \notag \\
&+ \tr \Big( \Psi(\tilde P_{t+1}, P_{t+1}^{11}) \cov(W_t^1)\Big).
\label{e_t}
\end{align}
\end{small}

\underline{\textbf{Calculating $\mathbb{T}(\emptyset)$:}}

Now that we have calculated $\mathbb{T}_1(\emptyset)$ and $\mathbb{T}_2(\emptyset)$, we use them to simplify $\mathbb{T}(\emptyset)$, and then by substituting it into the value function of \eqref{eq:DP_V} we get

\vspace{-4mm}
\begin{small}
\begin{align}
&V_{t}(x_{t}^0, m_{t}^0, \theta_{t}, \tilde z_{t})
=e_t + 
\min_{u^0_t \in \R^{d_U^0}, \bar q_t \in \bar{\mathcal{Q}}, \tilde q_t \in \tilde{\mathcal{Q}}(\theta_t) } 
\hspace{-2pt} \Big\lbrace \notag \\
&\hspace{1cm}   QF\left(H_t(m_t^{0},\emptyset), \bar S_t^{\theta_t}\right)
\nonumber\\
&\hspace{1cm} +  \sum_{m_t^1 \in \mathcal{M}^1} \tr\left(\tilde H_t(m_t^{0:1}) \cov(\tilde S_t^{\theta_t,m_t^1}) \right)
 \pi_{M^1}(m_t^1)
\Big\}.
\label{eq:DP_V_step1}
\end{align}
\end{small}
Note that $\bar S_t^{\theta_t} = \vecc(x_t^0, \mu(\theta_t),u^0_t, \bar q_t(1), \ldots, \bar q_t(\kappa^1))$ depends only on $u^0_t$ and $\bar q_t$. Furthermore, we have $\tilde S_t^{\theta_t,m_t^1} = \vecc(X^{\theta_t}, \tilde q_n(X^{\theta_t},m_t^1))$, and hence $\cov(\tilde S_t^{\theta_t,m_t^1})$ depends only on the choice of $\tilde q_t(\cdot,m_t^1)$. Consequently, in order to solve the optimization problem in the \eqref{eq:DP_V_step1}, we need to solve the two optimization problems

\vspace{-4mm}
\begin{small}
\begin{align}
 &\min_{u^0_t \in \R^{d_U^0}, \bar q_t \in \bar{\mathcal{Q}}} 
QF \Big(H_t(m_t^0,\emptyset), \bar S_t^{\theta_t} \Big),
\label{eq:DP_V_P1}
\\
&\min_{\tilde q_t \in \tilde{\mathcal{Q}}(\theta_t) }
\sum_{m_t^1 \in \mathcal{M}^1} \tr \Big(\pi_{M^1}(m_t^1)\tilde H_t(m_t^{0:1}) \cov(\tilde S_t^{\theta_t,m_t^1}) \Big)
\label{eq:DP_V_P2}
\end{align}
\end{small}
Since $H_t(m_t^0,\tilde z_t)$ is PD, it follows from \cite[Lemma 4]{Ouyang_Asghari_Nayyar:2016} that the optimal solution of \eqref{eq:DP_V_P1} is given by \eqref{eq:opt_ubar_mode_empty} and

\vspace{-4mm}
\begin{small}
\begin{align}
\label{eq:DP_V_part1}
&\min_{
\substack{u^0_t \in \R^{d_U^0}, \bar q_t \in \bar{\mathcal{Q}}}}
QF \Big(H_t(m_t^0,\emptyset), \bar S_t^{\theta_t} \Big) \notag \\
&= 
 QF \Big(P_t(m_t^0,\emptyset), \vecc \big(x^0_t,\mu(\theta_t) \big) \Big).
\end{align}
\end{small}
Furthermore, $P_t(m^0,\emptyset)$ is PSD because it is Schur complement of $H_t(m^0,\tilde z)$ with  respect to $H_t^{UU}(m^0,\emptyset)$. Similarly, since $  \tilde H_t(m_t^{0:1})$ is also PD, from \cite[Lemma 4]{Ouyang_Asghari_Nayyar:2016}, the optimal solution of \eqref{eq:DP_V_P2} is given by \eqref{eq:opt_gamma_mode} and

\vspace{-4mm}
\begin{small}
\begin{align}
& \min_{\tilde q_t \in \tilde{\mathcal{Q}}(\theta_t) }
\sum_{m_t^1 \in \mathcal{M}^1} \tr \Big( \pi_{M^1}(m_t^1)\tilde H_t(m_t^{0:1}) \cov(\tilde S_t^{\theta_t,m_t^1}) \Big)
 \notag \\
 & \sum_{m_t^1 \in \mathcal{M}^1}
\min_{\substack{\tilde q_t(\cdot,m_t^1): \\ \int \tilde q_t(x_t^1,m_t^1) \theta_t(dx_t^1) =0}} 
\hspace{-8mm}
 \tr \Big( \pi_{M^1}(m_t^1)  \tilde H_t(m_t^{0:1}) \cov(\tilde S_t^{\theta_t,m_t^1}) \Big)
 \notag \\
&= \tr \Big( \tilde{P}_t(m_t^0,\emptyset) \cov(\theta_t) \Big).
\label{eq:DP_V_part2}
\end{align}
\end{small}
Note that from \eqref{eq:opt_gamma_mode}, we have $\int \tilde q_t(x_t^1,m_t^1) \theta_t(dx_t^1) =0$ for all $m_t^1 \in \mathcal{M}^1$ and hence, $\tilde q_t \in \mathcal{\tilde Q}(\theta_t)$. Furthermore, $\tilde P_t(m^0,\emptyset)$ is PSD because it is convex combination of Schur complements of $\tilde H_t(m^0,\emptyset)$ with  respect to $\tilde H_t^{U^1U^1}(m^0,\emptyset)$.

Finally, substituting \eqref{eq:DP_V_part1} and \eqref{eq:DP_V_part2} into \eqref{eq:DP_V_step1}, we observed that $V_t$ has the form described in \eqref{eq:Vt_mode}. This completes the proof of the induction step for the case $\tilde z_t = \emptyset$ and the proof of the theorem for this case. Next we consider the case $\tilde z_t = l_t^1$.

\subsection{$\tilde z_t = l$ for some $l \in \mathcal{M}^1$}
This corresponds to the case that $\gamma_t =1$ from \eqref{Model:channel}. Hence, in this case $z_t = x_t^1$ and from \eqref{belief_theta_gamma_1}, $\theta_t = \delta(x_t^1)$. In the following we calculate $\mathbb{T}_1(l)$ and $\mathbb{T}_2(l)$.

\underline{\textbf{Calculating $\mathbb{T}_1(l)$:}}

Remember that we defined $S_t^{\theta_t,m_t^1} = \vecc(x_t^0, X^{\theta_t},u^0_t, \bar q_t(m_t^1)+ \tilde q_t(X^{\theta_t},m_t^1) )$ where $X^{\theta_t}$ is a random vector with distribution $\theta_t$. In this case, we have

\vspace{-4mm}
\begin{small}
\begin{align}
\label{mean_S_t_non_empty}
&\ee[S_t^{\delta(x_t^1),l}] = \vecc(x_t^0, x_t^1,u^0_t, \bar q_t(l)), \\
&\cov(S_t^{\delta(x_t^1),l}) =  \mathbf{0}.
\label{cov_S_t_non_empty}
\end{align}
\end{small}
According to \eqref{cost_structure} and \eqref{cost_of_new_problem},

\vspace{-4mm}
\begin{small}
\begin{align}
\mathbb{T}_1(l_t^1) &= 
\ee \Big[ c_t^{prs}(x_t^0, X_t^{1}, m_t^0, M_t^{1}, u^{prs}_t)  \Big| x_t^0, m_t^0, \delta(x_t^1), l, u_t^{prs} \Big]  \notag \\
&=  \ee \Big[QF\left(C_t(m_t^{0},l_t^1), S_t^{\delta(x_t^1),l} \right) \Big]
\notag \\
&= QF\left(C_t(m_t^{0},l), \ee[S_t^{\delta(x_t^1),l}]\right) \notag \\
&+  \tr\left(C_t(m_t^{0},l)\cov(S_t^{\delta(x_t^1),l}) \right) \notag \\
& = QF \Big(C_t(m_t^{0},l),  \vecc \big(x_t^0, x_t^1,u^0_t, \bar q_t(l) \big) \Big),
\label{T_1_non_empty}
\end{align}
\end{small}
where the second equality is true because of \eqref{eq:Vinduction_part1_first}.

\underline{\textbf{Calculating $\mathbb{T}_2(l)$:}}

Let $Y_t^{\delta(x_t^1),l}: = [D(m_t^0,l)]_{1\bullet} S_t^{\delta(x_t^1),l} + W_t^1$. Then, from \eqref{eq:psit} in Lemma \ref{lm:beliefupdate}, we know that $Y_t^{\delta(x_t^1),l}$ has distribution $\psi_t^{\triangleleft}(l, \delta(x_t^1))$ and we have,

\vspace{-4mm}
\begin{small}
\begin{align}
\label{mean_alpha_1_non_empty}
&\mu \big (\psi_t^{\triangleleft}(l, \delta(x_t^1)) \big) = \ee[Y_t^{\delta(x_t^1),l}] =
[D(m_t^{0},l)]_{1\bullet} \ee[S_t^{\delta(x_t^1),l}] ,   \\
&\cov\big( \psi_t^{\triangleleft}(l, \delta(x_t^1)) \big) = \cov(W_t^1).
\label{cov_alpha_1_non_empty}
\end{align}
\end{small}
Furthermore, in this case from \eqref{mean_alpha_0} and \eqref{cov_alpha_0}, 

\vspace{-4mm}
\begin{small}
\begin{align}
\label{mean_alpha_0_non_empty}
&\mu(\alpha_t^0) = [D(m_t^{0},l)]_{0\bullet} \ee[S_t^{\delta(x_t^1),l}],  \quad \cov(\alpha_t^0) = \cov(W_t^0).
\end{align}
\end{small}
Now, considering \eqref{mean_alpha_1_non_empty}, \eqref{cov_alpha_1_non_empty}, and \eqref{mean_alpha_0_non_empty},  
$\mathbb{T}_2(l)$ from \eqref{eq:Vterm_5_mode} can be calculate as follows,

\vspace{-4mm}
\begin{small}
\begin{align}
\mathbb{T}_2(l) &= QF \Big(D(m_t^{0},l)^{\tp} \Pi(P_{t+1}) D(m_t^{0},l), \ee[S_t^{\delta(x_t^1),l}] \Big) + e_t,
\label{T_6_non_empty}
\end{align}
\end{small}
where $e_t$ is as described in \eqref{e_t}.

\underline{\textbf{Calculating $\mathbb{T}(l)$:}}

Now that we have calculated $\mathbb{T}_1(l)$ and $\mathbb{T}_2(l)$, we use them to simplify $\mathbb{T}(l)$,
 and then by substituting it into the value function of \eqref{eq:DP_V} we get

\vspace{-4mm}
\begin{small}
\begin{align}
&V_{t}(x_{t}^0, m_{t}^0, \theta_{t}, \tilde z_{t})
=e_t +  \min_{u^0_t \in \R^{d_U^0}, \bar q_t \in \bar{\mathcal{Q}}, \tilde q_t \in \tilde{\mathcal{Q}}(\theta_t) } 
\hspace{-2pt} \Big\lbrace  
\nonumber\\
&QF \Big(H_t(m_t^0,l) \vecc \big(x_t^0, x_t^1,u^0_t, \bar q_t(l) \big) \Big) 
\Big\}.
\label{eq:DP_V_with_T_non_empty}
\end{align}
\end{small}
Note that the term inside the minimization does not depend on function $\tilde q_t$ and also it does not depends on $\bar q_t(m_t^1)$ for all $m_t^1 \in \mathcal{M}^1 \setminus \{l\}$. Hence, they can be chosen arbitrarily or set to be zero as described in \eqref{eq:opt_qbar_mode_l} and \eqref{eq:opt_gamma_mode}. Since $H_t(m_t^0,l)$ is PD, it follows from \cite[Lemma 4]{Ouyang_Asghari_Nayyar:2016} that the optimal solution of \eqref{eq:DP_V_P1} is given by \eqref{eq:opt_ubar_mode} and

\vspace{-4mm}
\begin{small}
\begin{align}
\label{eq:DP_V_part1_non_empty}
&\min
QF \Big(H_t(m_t^0,l) \vecc \big(x_t^0, x_t^1,u^0_t, \bar q_t(l) \big) \Big)  \notag \\
& = 
 QF \Big(P_t(m_t^0,l), \vecc \big(x^0_t,\mu(\delta(x_t^1)) \big) \Big).
\end{align}
\end{small}
Finally, substituting \eqref{eq:DP_V_part1_non_empty} into \eqref{eq:DP_V_with_T_non_empty}, we observed that $V_t$ has the form described in \eqref{eq:Vt_mode}. This completes the proof of the induction step for the case $\tilde z_t = l$ and the proof of the theorem for this case. 

\section{Proof of Theorem \ref{thm:opt_strategies}}
\label{Proof_Thm_optimal_strategies}
Let $\hat x_t^1$ be the estimate (conditional expectation) of $X_t^1$ based on the common information $h^0_t$. Then, for any realization of the common belief $\theta_t$, $\hat x_t^1 = \mu(\theta_t)$.
To show \eqref{eq:estimator_0} and \eqref{eq:estimator_t}, note that 
at time $t=0$, for any realization $h^0_t$ of $H^0_t$,

\vspace{-4mm}
\begin{small}
\begin{align}
&\hat x_0^1 = \mu(\theta_0) = \int y \theta_0(dy) \notag \\
&= 
\Big \lbrace \begin{array}{ll}
\int y \pi_{X_0^1}(dy) = \mu(\pi_{X_0^1}) & \text{ if }z_{0}= \emptyset,\\
 \int y \ind_{\{y\}} (x_0^1)(dy) =  
 x_0^1 & \text{ if }z_{0} = x_0^1.
\end{array} 
\end{align}
\end{small}
Therefore, \eqref{eq:estimator_0} is true. Furthermore, at time $t+1$ and for any realization $h^0_{t+1}$ of $H^0_{t+1}$,
let $\theta_{t+1}$ be the corresponding common belief and $u^{prs*}_t = (u_t^{0*}, \bar q_t^*, \tilde q_t^*)$, then
\begin{align*}
\hat x_{t+1}^1 = \mu(\theta_{t+1}) = \int  y [\psi_t(\theta_t,u^{prs*}_t, x_t^0,m_t^0, \tilde z_t, z_{t+1})] (dy).
\end{align*}
If $z_{t+1} = x_{t+1}^1$, then $\hat x_{t+1}^1 = \int y \ind_{\{y\}} (x_{t+1}^1) (dy) = x_{t+1}^1.$
\\
If $z_{t+1} = \emptyset$, then,

\vspace{-4mm}
\begin{small}
\begin{align}
\hat x_{t+1}^1 &= \int y  \sum_{m_t^1 \in \mathcal{M}^1} \int \int  \mathds{1}_{\{y\}}\big(f_t^1(x_t^{0:1},m_t^{0:1}, u_t^{prs*}, w_t^1)\big)  \notag \\
& \times \theta_t(dx_t^1) \omega^{\tilde z_{t}}(m_t^1) \pi_{W_t^1}(dw_t^1)(dy)
\notag\\
&=\sum_{m_t^1 \in \mathcal{M}^1} \int \int  f_t^1(x_t^{0:1},m_t^{0:1}, u_t^{prs*}, w_t^1)  \notag \\
&\times \theta_t(dx_t^1) \omega^{\tilde z_{t}}(m_t^1)\pi_{W_t^1}(dw_t^1)  \notag \\
&=\sum_{m_t^1 \in \mathcal{M}^1}
[D(m_t^{0:1})]_{1\bullet} \vecc(x_t^0, \hat x_t^1, \bar u^{0*}_t, \bar q_t^{*}(m_t^1))
 \omega^{\tilde z_{t}}(m_t^1),
\label{proof:estimation}
\end{align}
\end{small}
where the third equality is true because 
\begin{small}
\begin{align*}
&\int  y  \mathds{1}_{\{y\}}\big(f_t^1(x_t^{0:1},m_t^{0:1}, u_t^{prs*}, w_t^1)\big)   dy 
\notag \\&
= f_t^1(x_t^{0:1},m_t^{0:1}, u_t^{prs*}, w_t^1),
\end{align*}
\end{small}
Furthermore, the last equality of \eqref{proof:estimation} is true because $\tilde q_t \in \mathcal{\tilde Q}(\theta)$ and $W_t^1$ is a zero mean random vector.
Therefore, \eqref{eq:estimator_t} is true and the proof is complete.

\end{document}